\newif\ifdouble
\doublefalse

\ifdouble
\documentclass[]{IEEEtran}
\else
\documentclass[12pt, draftclsnofoot, onecolumn]{IEEEtran}
\fi

\usepackage{amsthm,enumitem}

\usepackage[pdftex]{graphicx}
\graphicspath{{./gfx/jpeg/}{./gfx/pdf/}{./gfx/pdf/QueueingSimRes/}{./gfx/pdf/QueueingSimRes/timeDepenendetModel/}{./}}
\DeclareGraphicsExtensions{.pdf,.jpeg,.png,.jpg}
\usepackage{mwe,tikz}

\usepackage[cmex10]{amsmath}
\usepackage{amssymb,dsfont,upgreek,color,verbatim}
\usepackage{caption}
\usepackage{subcaption}
\usepackage{lineno,hyperref}
\modulolinenumbers[5]

\usepackage{url}
\theoremstyle{plain}

\newtheorem{theorem}{Theorem}
\newtheorem{Definition}{Definition}
\newtheorem{lemma}{Lemma}
\newtheorem{remark}{Remark}
\newcommand{\eqdef}{=:}
\makeatletter
\newcommand{\rmnum}[1]{\romannumeral #1}
\newcommand{\Rmnum}[1]{\expandafter\@slowromancap\romannumeral #1@}

\newcommand{\FMOR}[1]{\textbf{\textcolor{red}{***FIXME ORI #1 ***}}}

\makeatother
\begin{document}


\title{Distributed Scheduling in Time Dependent Environments: Algorithms and Analysis} 

\author{Ori Shmuel, Asaf Cohen and Omer Gurewitz\\
Ben-Gurion University of the Negev, Israel\\
Email: \{shmuelor,coasaf,gurewitz\}@bgu.ac.il
}


\maketitle
\begin{abstract}

Consider the problem of a multiple access channel in a time dependent environment with a large number of users. In such a system, mostly due to practical constraints (e.g., decoding complexity), not all users can be scheduled together, and usually only one user may transmit at any given time. Assuming a distributed, opportunistic scheduling algorithm, we analyse the system's properties, such as delay, QoS and capacity scaling laws. Specifically, we start with analyzing the performance while \emph{assuming the users are not necessarily fully backlogged}, focusing on the queueing problem and, especially, on the \emph{strong dependence between the queues}. We first extend a known queueing model by Ephremides and Zhu, to give new results on the convergence of the probability of collision to its average value (as the number of users grows), and hence for the ensuing system performance metrics, such as throughput and delay. This model, however, is limited in the number of users one can analyze. We thus suggest a new model, which is much simpler yet can accurately describes the system behaviour when the number of users is large. 

We then proceed to the analysis of this system under the assumption of time dependent channels. Specifically, we assume each user experiences a different channel state sequence, expressing different channel fluctuations (specifically, the Gilbert-Elliott model). The system performance under this setting is analysed, along with the channel capacity scaling laws.

\end{abstract}


\section{Introduction}

In recent years, connectivity of everyday objects, which are often equipped with computing and networking capabilities, is becoming attractive and in some cases even necessary. These next generation information technologies form the new field of Internet of Things (IoT). A key aspect of such technologies is to support data transfer from a \emph{huge number} of nodes in the network (e.g., sensors), in order to provide novel applications. For example, many cities today provide smart city technologies such as smart metering, surveillance and security, infrastructure management, city automation, and eHealth. All these applications require data transfer from a huge number of sensors/devices, while sharing a common channel or infrastructure.

Due to this rapid growth in the number of users/devices, which rely on the wireless connectivity to a single gateway, we expect to have extremely large amount of traffic, data or control
, and since not all devices can transmit simultaneously, some channel access mechanism is required to coordinate between the transmitters efficiently. 

A common channel access for IoT and wireless sensors networks (WSN) is the random access mechanism \cite{pratas2012code,zhou2012contention,huang2013evolution}. Yet, as widely explored over the past few decades, such contention based access paradigms can result in low channel utilization due to collisions and mutual interference. Furthermore, since users are accessing the channel arbitrarily, regardless of their channel quality, a user with bad channel quality can capture the channel and transmit for a long duration (i.e., low rate), degrading the overall network performance \cite{heusse2003performance}. On the other hand, schedule-access-based protocols, allowing a scheduler to schedule users according to their current channel state, hence, exploit the multi-user diversity which is inherent to the wireless medium \cite{viswanath2002opportunistic,liu2001opportunistic}, may suffer from a large overhead and a complexity borden if the number of users is large.

Multi-user diversity gains are even more acute when utilizing multiple antennas both at the transmitter and the receiver. Many scheduling studies can be found for such Multiple-Input Multiple-Output (MIMO) technologies. For examples, in \cite{kim2005scheduling} and \cite{yoo2006optimality}, Zero-Forcing Beamforming (ZFBF) was investigated and user selection was performed in order to avoid interference among users' streams. Note, however, that such schedule-base schemes involve large channel state information exchange overhead, which hinders the large throughput gain \cite{bejarano2014mute}.

Considering the above, a suitable channel access scheme can be a distributed opportunistic threshold based algorithm, in which users can attempt transmission only if their channel state is above a threshold. On the one hand, such a scheme is opportunistic as it allows channel access only to users with advantageous channel conditions. On the other hand, it does not require extensive channel state information exchange, hence entails only a relatively low overhead. Several studies in the literature considered similar threshold-based opportunistic systems, e.g., \cite{qin2003exploiting,qin2006distributed} and \cite{kampeas2014capacity} for the homogeneous and non-homogeneous user cases, respectively. However, these studies investigated only the potential capacity or throughput gains of the mechanism, but have not analyzed other systems' metrics such as \emph{delay or buffer occupancy}. Moreover, they did not consider fully the very practical model of a time-dependent channel. \cite{qin2003exploiting} have a model with states but it's a bit different since one state is for no transmission at all and the other allow transmissions.

It is important to note that metrics such as delay or buffer occupancy are especially interesting in threshold-based algorithms, as one might suspect that a threshold-based algorithm will result in a long delay (until the threshold is exceeded) or in some kind of unfair scheduling (if a user exceeds the threshold more often than others).

\subsection{Main contributions}
In this work we address these concerns. Specifically, we consider a multi-user system comprising of $K$ users (devices) with a single antenna, wishing to communicate with a single gateway with multiple antennas. The channel access is governed by a threshold based random-access mechanism where each user
transmits packets in a first-in-first-out (FIFO) manner, such that packets arriving to a user which is already busy with a pending packet transmission wait for their turn to be transmitted. 
Accordingly, each user maintains a queue in which it stores packets waiting for transmission. Note that due to the shared medium the \emph{queues at different users are tightly correlated}. 
Furthermore, we assume \emph{time dependent channels}, that is, the channel state users experience is not \emph{i.i.d.}, and depends on the previous channel conditions. 

We provide analytical models and closed formulas to determine important properties of such systems. Specifically, the contributions of this work are divided into two main threads which, together, gives a complete understanding of the system. We start with exploring the \emph{performance} of the system by presenting approximate models for the systems' queues behavior. We first modify the model presented in \cite{ephremides1987delay}, while suiting it to our setting. Numerical results and interesting observations are presented for this model. Then, we present a simpler model which, in contrast to the former, is able to describe the system's behavior when the number of users (hence, queues) is large. In this work-flow, the first model essentially emphasizes the difficulty in tracking such a complex and \emph{dependent} system, even for time independent channels, while the second model, which is simpler, provides a good answer to this problem by decoupling the dependency between the queues. 

We then present a third approximation model, in which we tie the second approximation model with the assumption of time dependent channels. Note that this knot is not possible using the techniques in \cite{ephremides1987delay}, and is enabled by our simpler decoupling technique.  Specifically, we assume that each user experiences \emph{a time varying channel} modeled as a Good-Bad channel (the Gilbert-Elliot model \cite{gilbert1960}) which reflects the time varying channel distribution. We thus suggest a time dependent queue model for our multi-user Multiple Access Channel (MAC) system, which shows very good agreement with simulations results.

Finally, after assuming time dependent channels, we study the \emph{capacity} while considering both the centralized and the threshold-based distributed scheduling algorithms. We achieve closed analytic expressions, in two different ways, for the channel capacity scaling laws. Specifically, we use two statistical tools, Extreme Value Theory (EVT) and Point Process Approximation (PPA). These enable us to examine the limit distribution of the system's throughput. 

    The rest of the paper is organized as follows. Section \ref{sec-model description} describes the model and basic assumptions for this work. In Section \ref{sec-Performance analysis}, we examine the performance metrics of the queues under this algorithm, e.g., delay and buffer occupancy. Models for the time-independent and time-dependent scenarios are presented along with numerical results. In Section \ref{sec-Scaling Law Under Time Dependent Channel}, we focus on time-dependent channels and their \emph{asymptotic capacity} under centralized and distributed algorithms. Section \ref{sec-conclusion} concludes this work.

\section{System Model and Assumptions}\label{sec-model description}
	
	We consider an uplink system with $K$ independent users and one base station. We assume a slotted system in which the time axis is divided into fixed-length intervals, referred to as time slots or simply slots. Following a typical slotted system model, we assume that all nodes are synchronized and that transmissions can only start at slot boundaries. Obviously, simultaneous transmissions may result in collisions. We assume that each user maintains a queue in which the user stores the packets waiting for transmission. The packets are transmitted in a FIFO manner, in which each packet is repeatedly transmitted until received successfully by the base station. We focus our attention on the users' queues, as illustrated in Figure~\ref{fig-QueuingSystem}. We assume that the arrival process of new packets to each user's queue is characterized by a Poisson process. Accordingly, the users are not always backlogged. We further assume that all users are homogenous, thus, all users have the same arrival rate $\lambda$. At the beginning of each slot, each user estimates its own channel conditions, i.e., the expected achievable rate, and tests whether it exceeds a predefined threshold. Upon exceeding and having a packet to send, i.e., if the queue is not empty, the user attempts transmission. Throughout most of this paper the threshold value is set such that on average the probability for exceedance is $1/K$. Note that it was shown in \cite{qin2003exploiting} that under a \emph{fully backlogged} system this value is optimal. However, considering that users may not have packets to send at all times, the probability of $1/K$ is conservative. That is, it is possible that some slots will not utilized due to over-restrained transmission probability, or alternatively, that the arrival rate, which keeps the system stable, can be slightly higher if we allow users to be slightly more aggressive in their transmission attempts. We emphasise that the analysis presented in this work is correct for any threshold value and arrival rate as long as those maintain stability. We thus give evaluation of different exceedance probabilities via simulations and further discussion regarding the threshold in the sequel.


\begin{figure}[!t]
            \centering
            \includegraphics[width=2in]{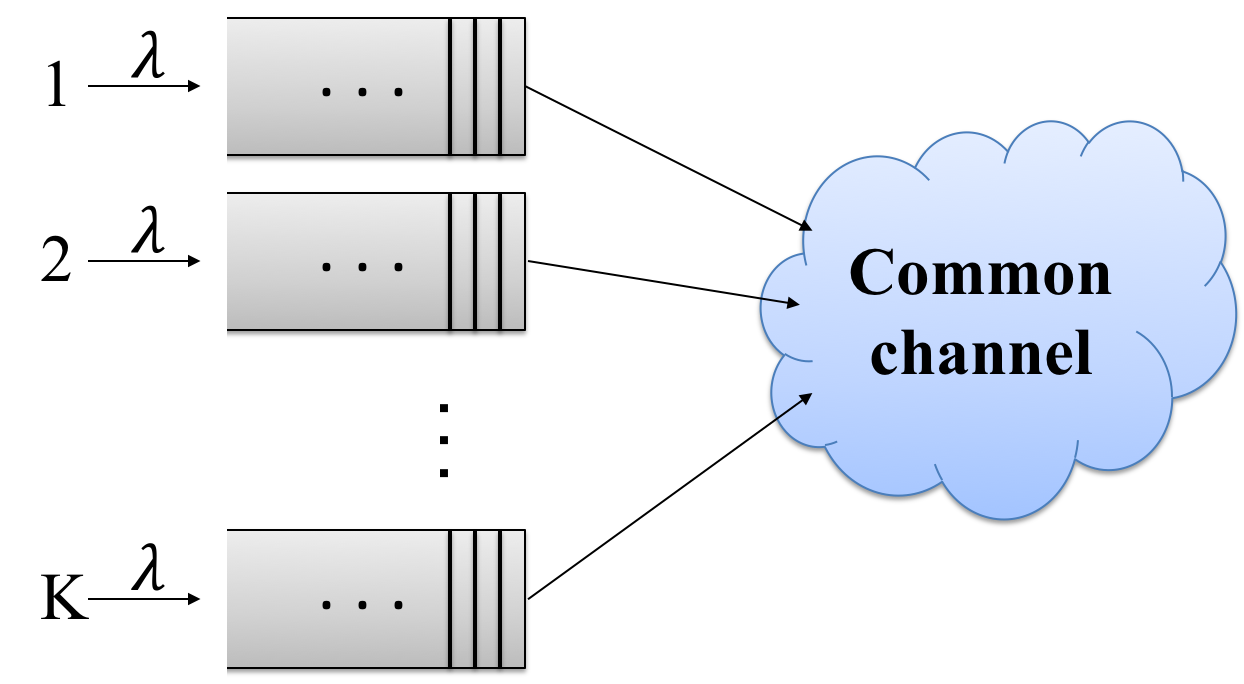}
            \caption[Model of queueing system]{System model. $K$ users access a common channel. Each user has a packet arrival process with rate $\lambda$.}
            \label{fig-QueuingSystem}
    \end{figure}

	We present several approximate models. Thus, new packets may enter the system at any given (continuous) point on the time axis or at the beginning of a slot, depending on the approximation model used. The slot size is set to encompass a single packet transmission at a rate which corresponds to the threshold. Note that since users are transmitting only while having an above-threshold capacity, and this capacity \emph{grows with the number of users} \cite{kampeas2014capacity}, the slots are expected to be small, specifically, we may neglect the transmission time in the analysis. In the first part of the paper, we assume that the achievable rates seen by each user at each time slot are \emph{i.i.d.} We extend these results to a setup in which these achievable rates are still identically distributed yet are not independent. In particular, we assume the channel of each user can alternate between different channel distributions (e.g., Good-Bad channel) according to the Gilbert-Elliot model. A full description will be provided in the sequel.

    We define the service time as the time from the moment a packet becomes first in queue, until it is successfully transmitted. Hence, the service time of the packets is composed of the waiting time \emph{for transmission} and the transmission time.

\section{Performance analysis}\label{sec-Performance analysis}

	Our system consists of $K$ queues, each with an independent Poisson arrival process, and a common server (i.e., the communication channel, Figure~\ref{fig-QueuingSystem}). A user will attempt transmission only when it is backlogged and the expected transmission rate for the next slot is above a threshold. We assume that any simultaneous transmissions will fail, i.e., no capture. The challenge in analyzing such a queueing system lies in the strong interdependence between the queues. Specifically, the user's collision or successful transmission probabilities depend on which of the other users' queues is backlogged (if the threshold exceedance probability is the same for all the users, it depends on the number of backlogged queues). Several works in the literature explored the behavior of various systems with queues interdependency. Due to the interdependence between the queues, each of these works considered a different simplified mathematical model, which resulted in approximations for the required metrics. For example, an iterative approximation model was suggested in \cite{saadawi1981analysis}, utilizing decoupled Markov chains. A refined model was presented in \cite{ephremides1987delay}. In \cite{sidi1983two}, the mean delay was given for the case of two identical users, as well as an approximation for a larger population. Extension for slotted CSMA/CD model can be found in \cite{takagi1985mean} and bounds on the stability region for such system can be found in \cite{luo1999stability}.  
	
In the sequel, we investigate the behavior of the system. We start by presenting two different approximations which capture the system performance. The first is a modification of the model presented in \cite{ephremides1987delay} to fit a threshold-based system. Different analytical results for the system performance metrics are obtained. The second model uses different technique from the works mentioned above which simplifies the analysis greatly. We extend the analysis for time dependent channels for each user.

    \subsection[Queueing Approximate model \Rmnum{1} ]{Approximation using System and Users' State}\label{Approximate model 1}
  
  	Trying to model the system state as the number of pending packets in each one of the queues, even though straightforward conceptually, is intractable and hence impractical. Accordingly, as our first approximation, we suggest an analytically tractable simplification which relies on the decoupling of the Markov chain into two separate and decoupled chains. This approximation is an adaptation of Ephremides and Zhu's model presented in~\cite{ephremides1987delay}, which analyzed the slotted ALOHA system with a finite number of buffered nodes. The main difference between the two models is that in slotted aloha described in~\cite{ephremides1987delay}, the transmission scheme is "immediate first transmission" i.e., if user $i$ has an empty queue when a packet arrives, it transmits the packet instantaneously, and in the case of collision, it may transmit again with a retransmission probability $p_i$, while our model follows a "delayed first transmission" scheme, where a transmission (first time or a later one) is delayed until the user channel state is favorable, i.e., happens only when the threshold has been exceeded. 
     

As in~\cite{ephremides1987delay}, we assume that the arrival rate at user $i$ is $\lambda_i$, and the arrival processes are statistically independent between the users\footnote{This model is able to capture heterogeneous arrival rates. Hence, we keep the subscript $i$ in its description. When we turn to the performance analysis we assume that $\lambda_i=\lambda$ as described in the system model.}. Time is slotted and it takes exactly one slot to transmit one packet. We will assume that since the transmission rates are high (users are transmitting only when their expected transmission rate is above a threshold) slot duration is quite small. Furthermore, since we are mainly interested in stable systems, i.e., the inter arrival time is much larger then the slot duration, we will assume that the probability that user $i$ receives a new packet for transmission during any given slot is $\Delta\lambda_i$, where $\Delta$ is the slot duration. We consider this duration as a unit size. The probability that a user receives more than one packet per slot is, similar to other Poisson models, negligible (i.e., $o(\Delta\lambda_i)$). Let us denote by $p_i(n)$ the probability that user $i$'s achievable rate in slot $n$ is above the predefined threshold. Since our first model assumes identically distributed channel conditions per slot, we will omit the slot index $n$, i.e., the threshold exceedance probability will be denoted by $p_i$ for all $n$. Accordingly, user $i$ attempts to transmit the head-of-the-line packet in its queue (given that its queue is nonempty), with probability $p_i$. Adopting the model in~\cite{ephremides1987delay}, we define three states in which each user can be at at the beginning of a given slot, namely \emph{Idle}, \emph{Active} or \emph{Blocked}. 
   
The states are determined at the beginning of each slot but depend on the previous slot as well. Specifically, a user is in \emph{Idle} state in two situations: having an empty queue at the beginning of the current slot, i.e., no packet arrival in the preceding slot, or having one packet in the beginning of the current slot which arrived after the beginning of the preceding slot. A user is in Blocked state if it was backlogged (i.e., its queue was not empty) at the beginning of the previous slot yet it has not transmitted successfully during the previous slot. Note that not transmitting successfully means either there was an unsuccessful transmission attempt, or that there was no transmission attempt at all, i.e., the user's channel state was below the threshold. A user is in Active state if it is backlogged at the beginning of the current slot and the user has successfully transmitted a packet during the last slot. We emphasize that the \emph{Active} state is an auxillary state which is utilized in the performance analysis by distinguishing successful transmissions for backlogged users.

Due to the above user $i$'s transmission probability is:
    \begin{equation}\label{equ-probability for transmission queueing }
    p_i= \left\{
          \begin{array}{l l}
            \frac{1}{K}\lambda_i & \ \text{if $i$ is idle},\\
            \frac{1}{K} & \ \text{if $i$ is active or blocked}.\\
    \end{array} \right.
    \end{equation}

As previously mentioned, the state space of such a system, which incorporates both the status of each user and the number of pending packets in its queue is intractable. Accordingly, our approximation relies on the decoupling of the Markov chain into two interdependent yet separate chains, the system-status chain and the queue-length chain. The transition probabilities of each chain are tightly depended on the steady-state probabilities of the other chain and thus, all state equations for both chains must be solved simultaneously.

The system-status chain captures the state of each user at any given time. Hence, the status variable $\overline{S}$ consists of $K$ ternary variables, $S_1,S_2,...,S_K$, each of which indicates the status of the corresponding terminal. Namely, $S_i=\{0,1,2\}$ for Idle, Active and Blocked, respectively. Since the system can incorporate at most one active user at any given time (i.e., at most one user can transmit successfully in the previous slot), it can be shown by summing over all the states where no users are active plus all the possible states in which a single user is active, that the total number of states achievable is $2^{K-1}(K+2)$. The transition probabilities of the system-status chain, which are different from the one presented in \cite{ephremides1987delay}, are given in Appendix \ref{Appendix A}.  However, we emphasize two quantities which are required for the system-status chain transition probabilities calculation:
    \begin{equation}\label{equ-p(1|1) and p(0|2) definition}
      \begin{array}{l}
            P_i(1\mid 1) \triangleq P_r(\text{queue size $>1\mid$ user $i$ is active})\\
            P_i(0\mid 2) \triangleq P_r(\text{queue size $=1\mid$ user $i$ is blocked}).
          \end{array}
    \end{equation}
Note that these probabilities, which are utilized in the system-status chain solution, reflect the coupling between the two chains. That is, the steady-state probability distribution $P(\overline{S})$ relies on the queue-length chain steady state. Their calculation can be found in Appendix \ref{Appendix B}.


The queue-length Markov chain tracks both the status and queue length of each user, independently of the status and queue length of the other users. Specifically, the pair $(T_i,N_i)$ represents the state of user $i$, where $N_i$ denotes the total number of packets at queue $i$ and $T_i$ is an indicator variable, indicating whether the user is blocked or unblocked (active or idle status), denoted by $0$ and $1$, respectively. We further denote by $ \mathbf {\pi}(T_i,N_i)$ user $i$'s steady-state probabilities.


   The transition probabilities of the chain depend on the following average transmission success probabilities (assuming a user has a packet to send):
    \begin{equation}\label{equ-Average success probabilities model 1}
    \begin{array}{l}
            P_B(i) =  P_r(\text{success $\mid$ user $i$ is blocked})\\
            P_A(i) =  P_r(\text{success $\mid$ user $i$ is active})\\
            P_I(i) =  P_r(\text{success $\mid$ user $i$ is idle}),
    \end{array}
    \end{equation}
    where the averaging is performed over the status of the \emph{other users} (the status of the system). Namely, in order to calculate the probabilities \eqref{equ-Average success probabilities model 1}, one needs the stationary distribution of the system-status chain $P(\overline{S)}$. This leads again to the coupling of the two sets of equations. The calculations of these average success probabilities are presented in Appendix \ref{Appendix B}. The conditional moment generating function of the chain for user $i$ is defined as
    \begin{equation}\label{equ-conditional moment generating function}
      G_{T_i}^i(z)\triangleq \sum_{N_i=0}^{\infty} \pi(T_i,N_i)z^{N_i}, \ \ \ \ T_i=0,1.
    \end{equation}
  
     In order to find the steady state probabilities of the decoupled chains, which, as previously explained, must be solved simultaneously, we utilize an iterative process which is based on the Wegstein iteration method. Specifically, at each iteration we exploit the auxiliary quantities $P_i(1\mid 1),P_i(0\mid 2)$ computed in the previous iteration to compute $P_I(i),P_A(i),P_B(i)$, and vice versa. After achieving satisfying convergence, the performance metrics can be calculated using the steady state of the chains.
     
   We emphasize here that due to the "delayed first transmission" (which is the result of the threshold exceedance requirement) assumption, the transition probabilities of the system chain became quite complicated compared to the analytical derivation in \cite{ephremides1987delay}.

\subsubsection{Queueing performance analysis}\label{Approximate model 1 - Queueing performance analysis}

In this subsection, we analyze the delay a packet endures from the moment it is generated and arrives to a user's queue until it is successfully transmitted. This delay can be divided into three components: The queueing time to allow the packets already queued to be transmitted, denoted by $W_q(i)$; the head-of-the-line waiting time, which is the time elapsed since the packet became first in the queue until successful transmission, denoted by $W_s(i)$; the transmission duration which is exactly one slot. Each of them is approximated separately, and the sum of the three gives the total average delay (measured in slot duration):
    \begin{equation}\label{equ-Delay of user i model 1}
      D_i=W_q(i)+W_s(i)+1.
    \end{equation}
Our main contributions for the performance analysis are given in the following Theorem, which gives the metrics affected directly from the threshold-based scheduling scheme. That is, the probability for a  packet to be blocked is now affected also by the successful exceedance of the threshold and thus the service time is also affected. In addition, the probability for success given a transmission attempt, which was not calculated in \cite{ephremides1987delay}, is given and will be compared with the approximation model in the next Section.

\begin{theorem}\label{Approximate model 1 - performance metrics theorem}
Under the queueing model which follows the "delayed first transmission" scheme, the performance metrics, i.e., the head-of-the-line waiting time, the probability for a packet to be blocked and the probability for success given a transmission attempt are, 
\begin{equation}\label{equ-service time approximation model 1}
      W_s(i)=Pr(\text{blocked packet}) \cdot \frac{1}{P_B(i)},
\end{equation}
\begin{equation} \label{equ-probability to be blocked approximation model 1}
\begin{aligned}
    Pr(\text{blocked packet})&=\left(1-Pr(\text{unblocked packet})\right)\\
    &=1-\left( \frac{\pi(1,0)}{\lambda \pi(1,0) +(G^i_1(1)-\pi(1,0))}P_I(i)+ \right.\\ 
    &\quad \quad \quad \quad \quad \quad \quad \left. \frac{G^i_1(1)-\pi(1,0)}{\lambda \pi(1,0) +(G^i_1(1)-\pi(1,0))}P_A(i)\right),
\end{aligned}
\end{equation}
    and,
     \begin{equation}\label{equ-success probability model 1}
      p_{succ}(i)= \frac{P_A(i)(G_1^i(1)-\pi(1,0))+P_I(i)\pi(1,0)+P_B(i)G_0^i(1)}{\frac{1}{K}(1-(1-\lambda)\pi(1,0))}.
    \end{equation}
\end{theorem}

\begin{proof}
The head-of-the-line waiting time takes into account the time the packet spends at the head of the queue excluding the successful transmission slot. Accordingly, the head-of-the-line waiting time is zero both in the case of successful transmission upon arrival at an empty queue (i.e., successful transmission from an \emph{Idle} state) and in the case of a successful transmission of a backlogged packet, in the slot consecutive to its becoming the head of the queue (i.e., successful transmission from an \emph{Active} state). Consequently, the head-of-the-line waiting time is the time the packet spends in \emph{Blocked} state (it is zero for packets which did not pass through Blocked state). Since the number of slots until successful transmission while Blocked is a geometric random variable with mean equal to $1/P_B(i)$ we have Equation \eqref{equ-service time approximation model 1}.

The probability for a packet to be blocked (Equation \eqref{equ-probability to be blocked approximation model 1}) can be computed as the complement of the probability to be unblock (successful transmission without passing through Block state) which is essentially the probability of immediate successful transmission upon arrival of a packet (to the head of the line) which may be only in Idle or Active states. The terms before the success probabilities are the proportion of successful transmissions while in Idle and Active states, respectively.

The probability for success given a transmission attempt (Equation \eqref{equ-success probability model 1}), i.e., the user has packet to send, exceeds the threshold and successfully transmits can be computed as the general success probability regardless the users' state (Idle, Active, Blocked), divided by the probability for transmission attempt. Note that this event is the result of two independent events, the user channel norm exceeding the threshold and the transmission being successful. Note also that the user can be in Idle state and yet manage to transmit successfully, which can occur upon successful transmission of a packet upon its arrival to an empty queue.
\end{proof}

The remaining performance metric, i.e., the queueing time, is calculated, standardly, using Little's result, hence 
    \begin{equation}\label{equ-time in line approximation model 1}
      W_q(i)=\frac{L_i}{\lambda_i},
    \end{equation}
    where $L_i$ is the average queue length of a user (without considering the blocked head-of-line packet as part of the queue), which is given by \cite{ephremides1987delay}:
    \begin{equation}\label{equ-mean queue size model 1}
      L_i=\frac{\lambda_i^2\overline{\lambda}_i\overline{P}_I(i)}{(\overline{\lambda}_iP_B(i)-\lambda_i\overline{P}_A(i))(\overline{\lambda}_iP_B(i)-\lambda_i(P_I(i)-P_A(i)))},
    \end{equation}
    where $\overline{P}=1-P$.

    \subsubsection{Performance results}\label{Approximate model 1 - results}
    In order to gain some insight on the analytical results attained in the previous subsection, in this subsection we present some numerical results. In particular, we present the performance attained by the threshold based channel access mechanism under various metrics such as delay, average queue size and success probability. 
    
    The simulation and analytic calculations were performed under homogenous settings in which all users experience the same arrival process and channel distribution. Specifically, the arrival process was approximated by a Bernoulli process in which at each slot each user receives a new packet for transmission with probability $\lambda_i=\frac{\lambda}{K}$. The threshold was set such that the exceedance probability is $\frac{1}{K}$. Accordingly, each backlogged user examined its channel state at the beginning of each slot, and if its channel state was above a threshold it started transmission. Note that since the channel distribution was homogenous in time, the threshold exceedance process can also be viewed as Bernoulli process in which the transmission probability is $\frac{1}{K}$. Since, as explained earlier, the number of states in the system-status Markov chain grows exponentially with the number of users, for this model, the simulation and the numerical results were compared only for modest number of users and specifically only for $2$ to $10$ users. In section \ref{Approximate model 2} we present a model which is able to capture a much larger number of users.

We start by evaluating the effect of the number of users on the average queue length and on the delay, where we divide the delay to its different components, namely, the time in queue and service time. Figures \ref{fig-MeanQueueSize_2-10_0366}, \ref{fig-TimeInLine_2-10_0366} and \ref{fig-ServiceTime_2-10_0366} depict the performance metrics as a function of the number of users, where the total arrival rate is set close to the maximum value which still allows the system to be stable. The results depict very good match between the simulations and the analytic results given in \eqref{equ-mean queue size model 1}, \eqref{equ-time in line approximation model 1} and \eqref{equ-service time approximation model 1}. However, as mentioned earlier, this approximation can only be used for a small number of users.
    
Figures \ref{fig-MeanQueueSize_K=7}, \ref{fig-TimeInLine_K=7} and \ref{fig-ServiceTime_K=7} depict performance as a function of the total arrival rate, which asserts the following observations: First, for small and moderate values of the arrival rate one can see very good agreement between the analytic results and the simulation. Near the maximum value, however, around the instability region, the values start to diverge, especially in Figures \ref{fig-MeanQueueSize_K=7} and \ref{fig-TimeInLine_K=7}. These results are consistent with the results of \cite{ephremides1987delay}, which used small values of $\lambda_i's$ as well. The second observation is that as the number of users grows, one can see even stronger compliance with the analytical results. This suggests that this approximate model may be suitable for a large population. Unfortunately, as mentioned before, the calculation of the system steady state is intractable due to the exponential growth of the number of states with the number of users.

     \ifdouble
\else
    \begin{figure*}[!t]
        \centering
        \begin{subfigure}[b]{0.31\textwidth}
                \centering
                \includegraphics[width=\textwidth]{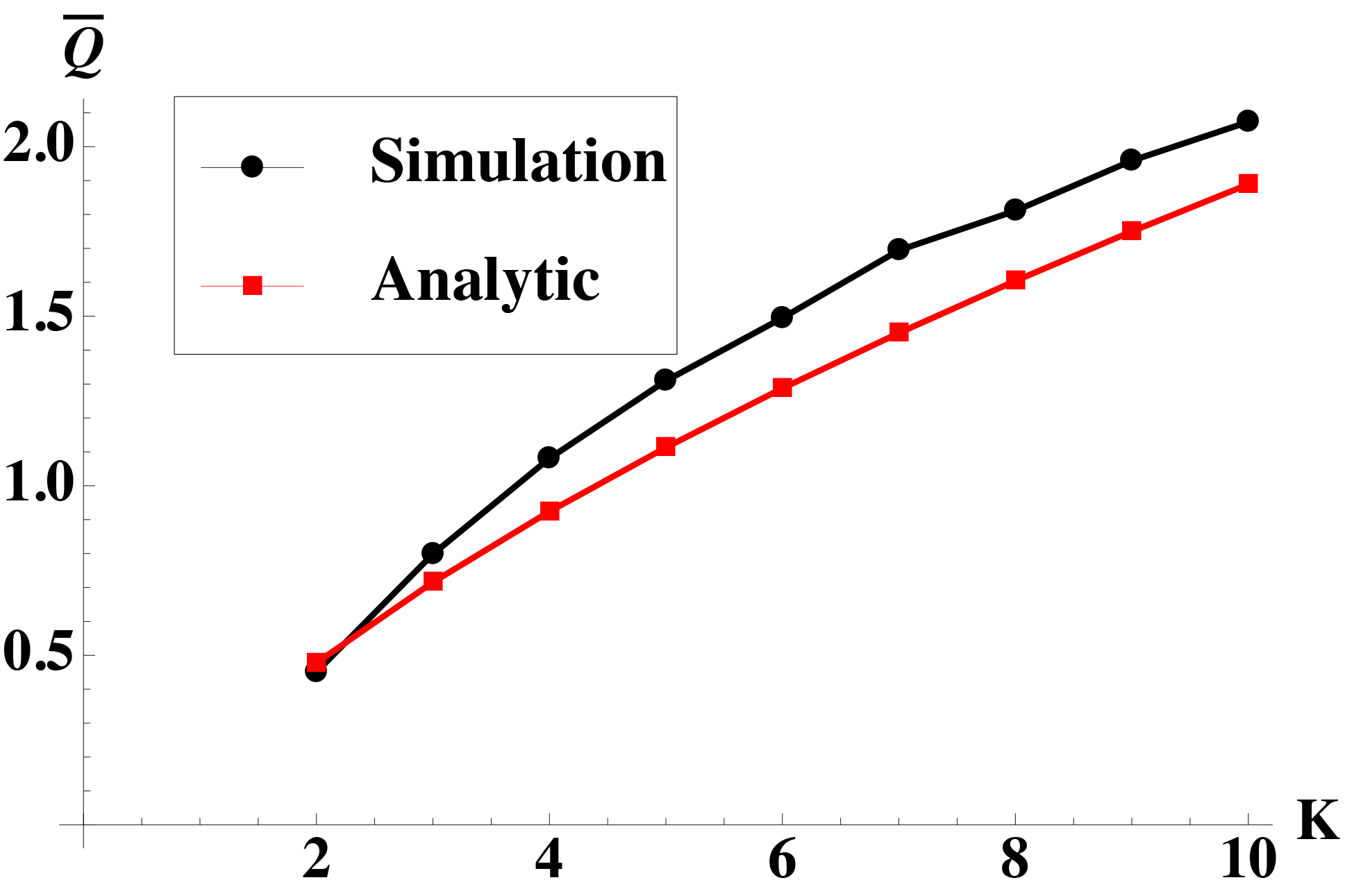}
                \caption{Mean queue size.}
                \label{fig-MeanQueueSize_2-10_0366}
        \end{subfigure}%
        \begin{subfigure}[b]{0.31\textwidth}
                \includegraphics[width=\textwidth]{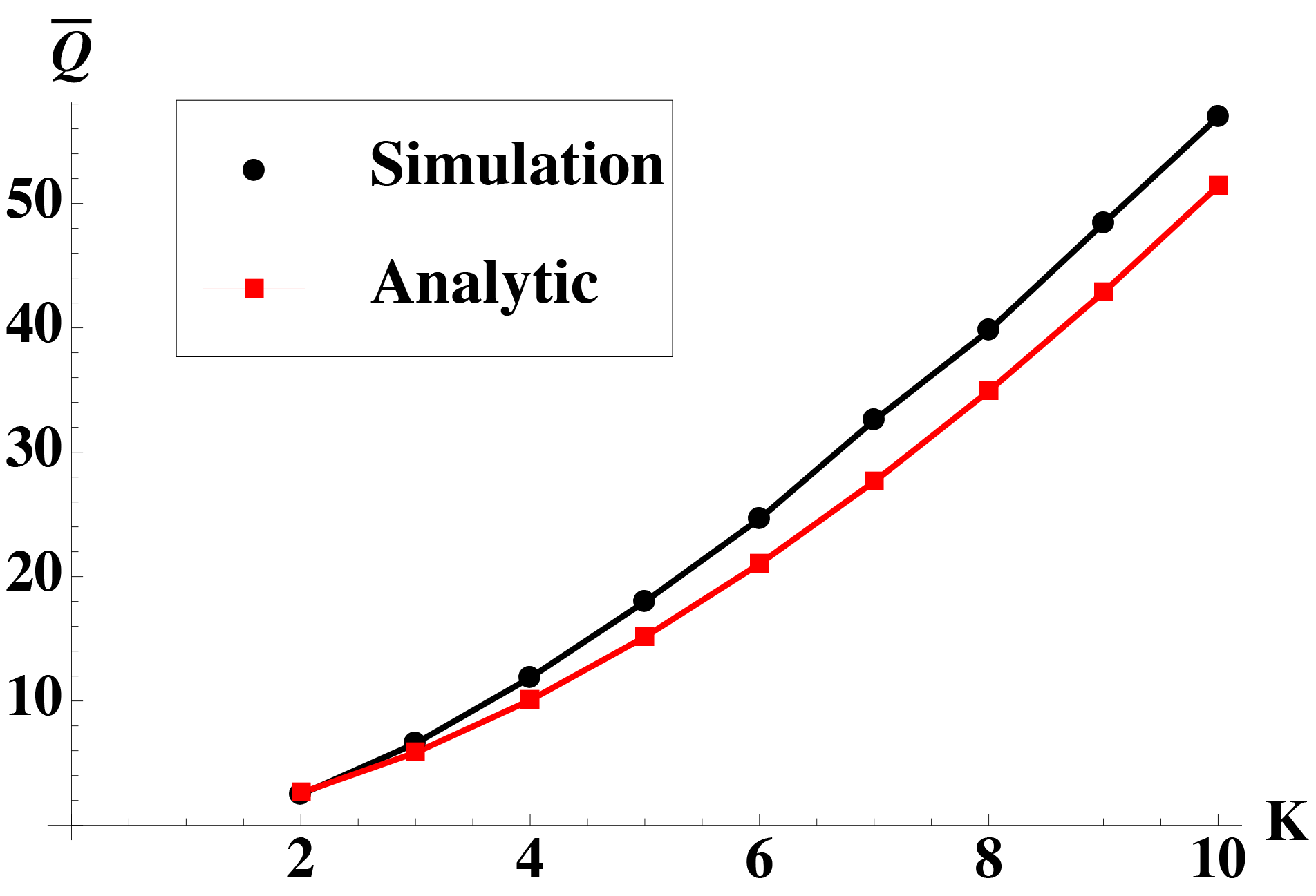}
                \caption{Time in line.}
                \label{fig-TimeInLine_2-10_0366}
        \end{subfigure}
        \begin{subfigure}[b]{0.31\textwidth}
                \includegraphics[width=\textwidth]{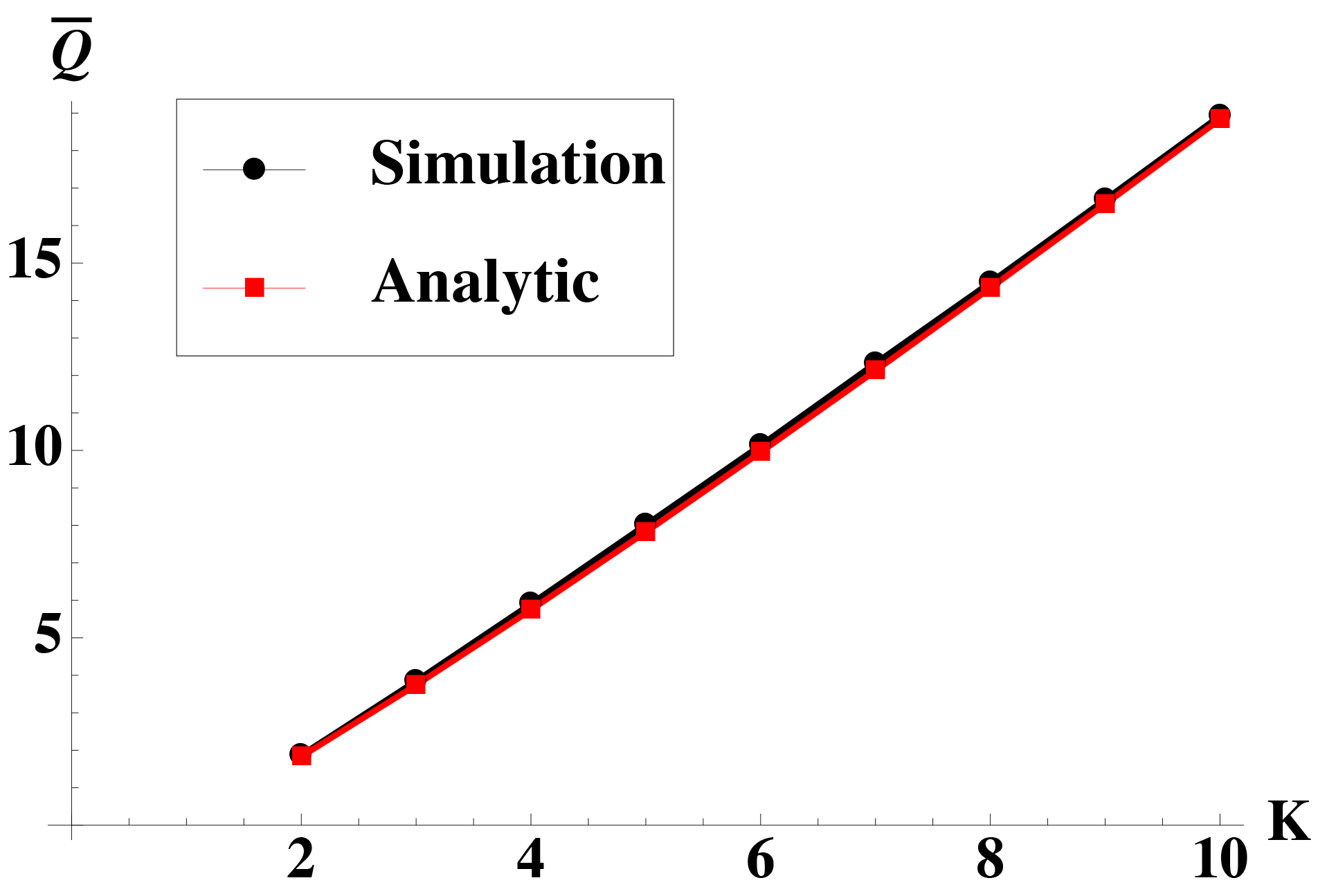}
                \caption{Service time.}
                \label{fig-ServiceTime_2-10_0366}
        \end{subfigure}
        
        \begin{subfigure}[b]{0.31\textwidth}
                \centering
                \includegraphics[width=\textwidth]{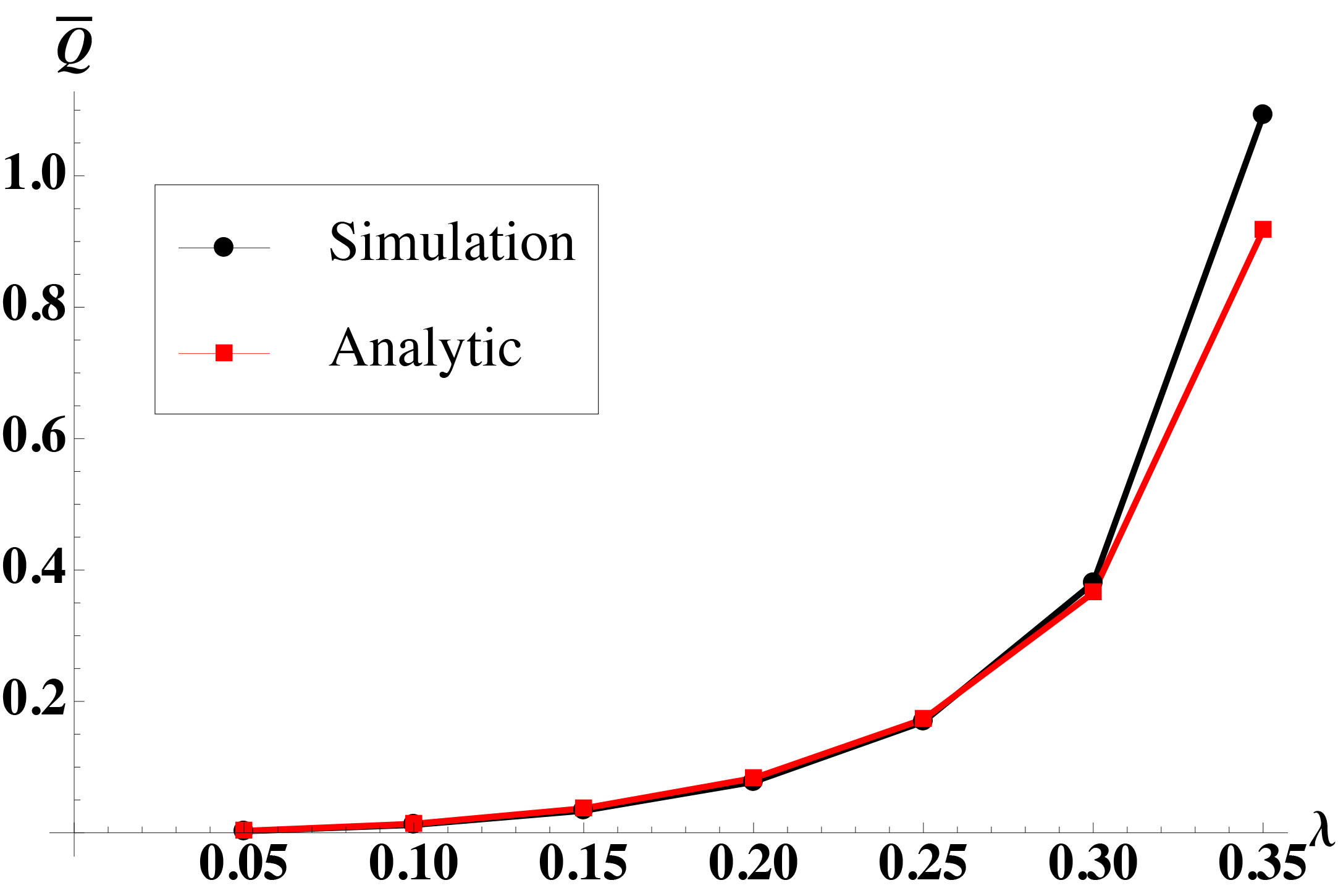}
                \caption{Mean queue size.}
                \label{fig-MeanQueueSize_K=7}
        \end{subfigure}%
        \begin{subfigure}[b]{0.31\textwidth}
                \includegraphics[width=\textwidth]{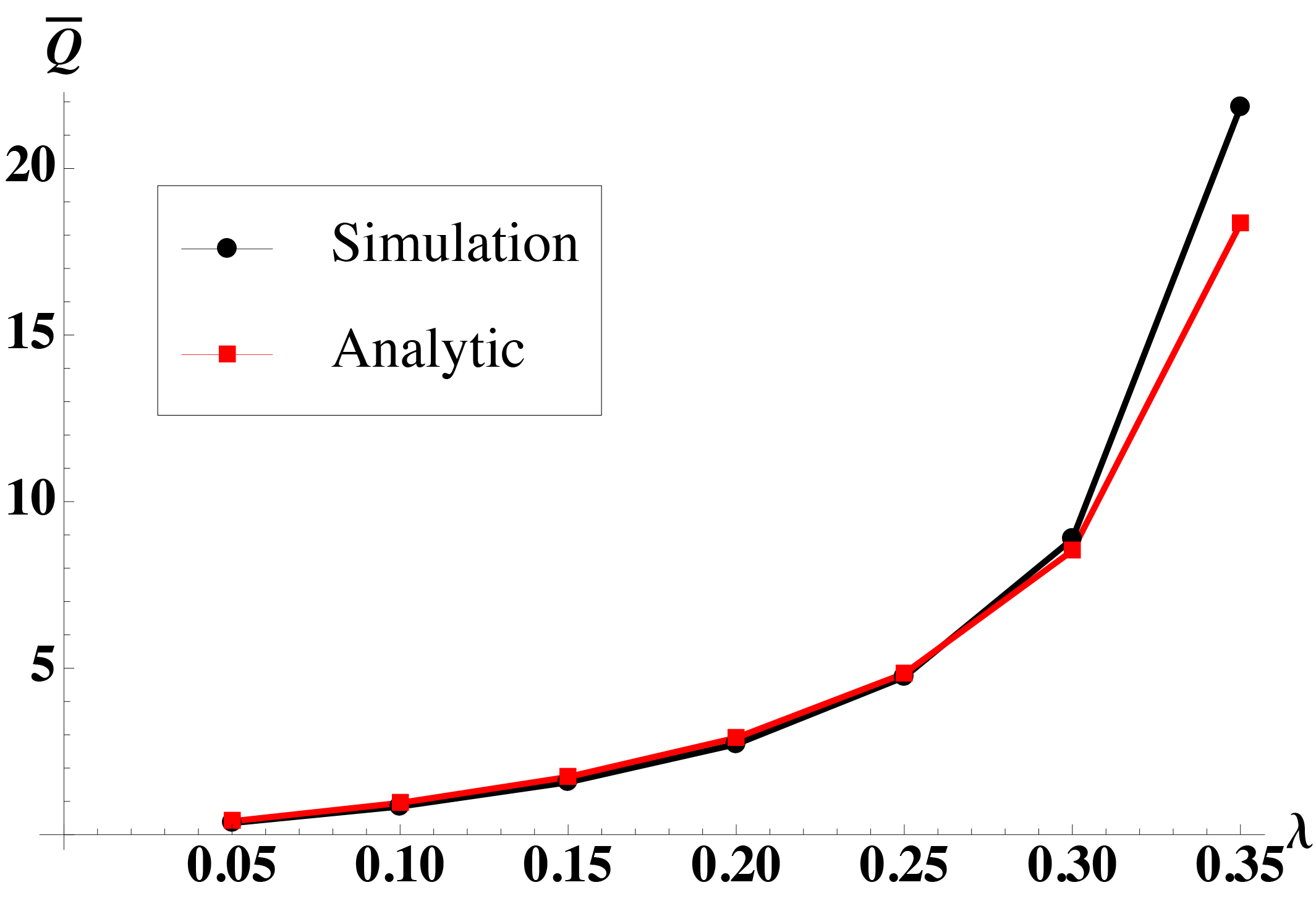}
                \caption{Time in line.}
                \label{fig-TimeInLine_K=7}
        \end{subfigure}
        \begin{subfigure}[b]{0.31\textwidth}
                \includegraphics[width=\textwidth]{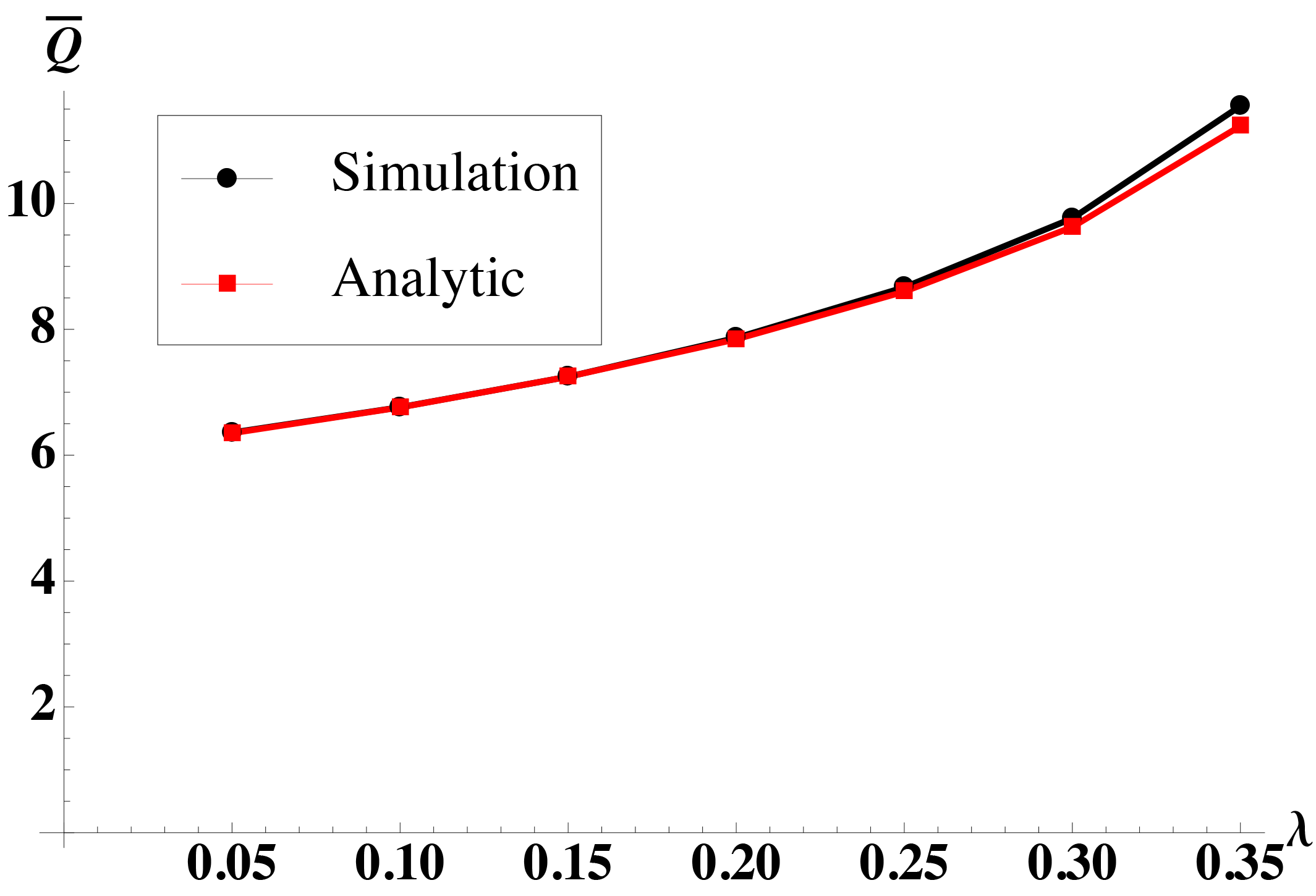}
                \caption{Service time.}
                \label{fig-ServiceTime_K=7}
        \end{subfigure}

        \caption[System performance time independent model]{Analytic results of the approximation using system and users' state compared with simulation for the system performances. The first raw is as a function of the number of users where the total arrival rate is $\lambda_T=\frac{1}{e}(1-0.001)$ and the second raw is as a function of the total arrival rate where the number of users is $K=7$. Figures (a,d) depict the mean queue size, (b,e) depict the time in line and (c,f) depict the service time. The red lines describes the anaclitic expressions \eqref{equ-mean queue size model 1}, \eqref{equ-time in line approximation model 1} and \eqref{equ-service time approximation model 1}, while the black lines describe the simulation results.}

        \label{fig-QueuingPerformance}
    \end{figure*}
    \fi

Next we compare the estimated success probability given a transmission attempt, $p_{succ}$ with simulation results. Figure~\ref{fig-SuccessProbability_2-10_0366_independent} clearly depicts that the analytical approximation described in \eqref{equ-success probability model 1} matches the simulation results with high accuracy. Even for a small number of seven users we have a difference of only 1.19\%. This highly accurate approximation is the foundation for a simplified model, described in the following section, which, in contrast to the model considered thus far, is able to capture the system's behavior for a large number of users with high accuracy.


    \subsection[Queueing Approximate model \Rmnum{2} ]{Approximation using Constant Collision Probability}\label{Approximate model 2}
    As described in the previous section, a user's success (or conversely, collision) probability depends on the other users' queues, hence on the system state. That is, different system states will result in different success (collision) probabilities. Trying to solve this set of equations is complicated even for a moderate number of users, all the more so for a large user population. In this section, we present a simpler approach, one which assumes that \emph{a user's collision probability is constant}. This will allow us to give close form results which are easy to calculate, and, as simulations depict, are very accurate for large and even moderate for population sizes.
    
    The key approximation method we adopt in this section is inspired by the mean field theory. Mean field theory studies the behavior of a large number of particles which interact. Specifically, when the number of particles is large, mean field approximation suggests an \emph{independent evolution of a certain particle relatively to others} by approximating the effect of all other particles on that particle by the averaged effect (essentially this is a concentration result). Ever since the mean field approximation was introduced in physics, it was adopted by various fields, including in the context of Markov process models for various dynamic systems. For example, in \cite{bordenave2012asymptotic}, mean field approximation was used to describe the stability region of the slotted Aloha paradigm. Therein, the authors proved that the distribution of a user's queue state is not affected by the other users when $K$ goes to infinity. Given that, the collision probability can be approximated by a fixed point equation, assuming the probability for empty queue is constant and is independent of the other users. Another seminal work which utilized a constant collision probability as the key approximation is \cite{bianchi2000performance}, denoted as the \emph{decoupling approximation}. In this work, a Markov model for the 802.11 back-off process was considered, with $n$ users competing on a shared medium. This decoupling assumption can be formally justified as a consequence of convergence to mean field, as the number of users goes to infinity. Several elaborations were made along with verification for the validity of the decoupling approximation \cite{malone2007modeling}, \cite{cho2012asymptotic}.  We also note that in \cite{sidi1983two}, independent M/M/1 queues were assumed, with a given distribution on the number of backlogged users. This resulted in an approximation with a \emph{varying} collision probability. 
Herein, and similar to the works mentioned above, we assume a constant collision probability which we denote by $p_{coll}$, i.e., we assume that given a transmission attempt, each user experiences a fixed collision probability regardless of the state or queues of the other users, and regardless of its own state. Even though the mean field theory mostly applies to large and complex stochastic models involving a large number of particles, we will show via simulation that our approach gives high accuracy even for a moderate number of users. We start by investigating the service time.

    \subsubsection{Service time analysis}
     The service time, which is the time from the moment a packet becomes first in queue, until it is successfully transmitted, depends on the rate at which a user exceeds the threshold and the probability of success (the probability that a collision did not occur). We note that the threshold exceedance process of each user can be modeled as a series of Bernoulli trials. We further note that this process converges to a Poisson process \cite{leadbetter1976weak}, especially if the slot duration is small compared to the time between threshold exceedances. Accordingly, we will approximate the time between threshold exceedances for each user as exponentially distributed with parameter $\uptau$. In the second part of this work, we will prove that indeed, as the number of users grows and the threshold exceedance probability decreases respectfully, convergence exists and exponential distribution between consecutive thresholds exceedances can be considered (see Section~\ref{subsec-Threshold_exceedance_process}).

Since we decoupled the queues, we can now assume that each user's queue behaves like an $M/M/1$ queue with a feedback loop. Specifically, packets enter each user's queue according to Poisson process with rate $\lambda_i$. The inter exceedance interval is exponentially distributed with mean $1/\tau$. Upon exceeding the threshold, a packet will be successfully transmitted (hence depart the queue) with probability $p_{succ}=1-p_{coll}$, and will need to be retransmitted with probability $p_{coll}$. Accordingly, the queue can be modeled as M/M/1 queue with exponentially distributed service time with parameter $(1-p_{coll}) \uptau$. The probability of an empty such $M/M/1$ queue, is:
    \begin{equation}\label{equ-the probability for empty queue M/M/1}
      P(Q=0)=1-\rho=1-\frac{\lambda}{(1-p_{coll}) \uptau},
    \end{equation}
    where $Q$ is the number of packets in the queue. We thus have the following lemma:

    \begin{lemma}\label{lem-Transmission success probability satisfies the equation - Memoryless arrival}
       Assume each user's queue is modeled as an $M/M/1$ queue with a feedback loop, with an arrival rate $\lambda$, exceedance rate $\uptau$ and constant collision probability $p_{coll}$. Then, the probability for collision $p_{coll}$ satisfies the equation
       \begin{equation}\label{equ-Transmission success probability satisfies the equation - Memoryless arrival before 1}
         p_{coll}=1-e^{-\frac{\lambda}{(1-p_{coll}) \uptau}} \cdot \left(1+o(1)\right).
       \end{equation}
     Hence, for large enough $K$ we have the following closed form for $p_{coll}$
     \begin{equation}\label{equ-Transmission success probability satisfies the equation - Memoryless arrival}
         p_{coll}=1-e^{-\frac{\lambda}{(1-p_{coll}) \uptau}}.
       \end{equation}
    \end{lemma}
    \begin{proof}
     Since random selection (with probability $(1-p_{coll})$) is performed on the threshold exceedance Poisson process, the outcome successful transmission process is also a Poisson process with rate $(1-p_{coll})  \uptau$, which leads to an exponential service time with parameter $(1-p_{coll})  \uptau$. Let us examine the probability $p_{coll}$.

Since the channel quality and the queue length are independent, the probability that a user will attempt transmission, equals the probability that the user is backlogged (i.e., its buffer is not empty) times the probability that its expected rate is above a threshold. Specifically,

     \begin{equation}
       P(C_i>u,Q_i>0)=P(C_i>u)P(Q_i>0)=\frac{1}{K}\cdot\frac{\lambda}{ (1-p_{coll})\uptau},
     \end{equation}
     On the other hand, given that a user has transmitted, i.e., its expected rate exceeds $u$ and its queue is not empty, its collision probability, $p_{coll}$, is the probability that among all the other users at least one other user will attempt transmission, i.e., at least one other user is backlogged and its expected rate exceeds $u$. This is one minus the probability that no other user has attempted transmission. Thus
     
      \ifdouble
\else
      \begin{equation*}
      \begin{aligned}
          p_{coll}=&\sum_{i=1}^{K-1}  \binom {K-1} {i} \left(\frac{1}{K} \frac{\lambda}{\uptau (1-p_{coll})} \right)^{i}\left(1-\frac{1}{K} \frac{\lambda}{\uptau (1-p_{coll})} \right)^{K-1-i} \\
          &= 1-\left( 1-\frac{1}{K} \frac{\lambda}{\uptau (1-p_{coll})} \right)^{K-1}\\
          &= 1- e^{(K-1)\ln{\left(1-\frac{\lambda}{(1-p_{coll}) \uptau K}\right)}}\\
          &= 1- e^{(K-1)\left(-\frac{\lambda}{(1-p_{coll}) \uptau K}-O\left(\frac{1}{K^2}\right)\right)}\\
          &= 1- e^{-\frac{\lambda}{(1-p_{coll}) \uptau}-O\left(\frac{1}{K}\right)} \cdot e^{\frac{\lambda}{(1-p_{coll}) \uptau K}+O\left(\frac{1}{K^2}\right)},\\
          &= 1- e^{-\frac{\lambda}{(1-p_{coll}) \uptau}} \cdot e^{\frac{\lambda}{(1-p_{coll}) \uptau K}+O\left(\frac{1}{K}\right)},\\
          &= 1- e^{-\frac{\lambda}{(1-p_{coll}) \uptau}} \cdot \left(1+O\left(\frac{\lambda}{(1-p_{coll}) \uptau K}\right)\right)e^{O\left(\frac{1}{K}\right)},\\     
          &= 1- e^{-\frac{\lambda}{(1-p_{coll}) \uptau}} \cdot \left(1+o(1)\right)e^{O\left(\frac{1}{K}\right)},\\ 
          &= 1- e^{-\frac{\lambda}{(1-p_{coll}) \uptau}} \cdot \left(1+o(1)\right),\\         
      \end{aligned}
     \end{equation*}
     \fi
     which completes the proof.
    \end{proof}

    Equation~\eqref{equ-Transmission success probability satisfies the equation - Memoryless arrival} is an implicit equation and a numerical method is needed in order to find the value of $p_{coll}$. In Figure \ref{fig-SuccessProbability_2-10_0366_independent}, we depict the numerical calculation of the success probability (the blue line), as given in \eqref{equ-Transmission success probability satisfies the equation - Memoryless arrival}, compared with the approximation derived in the previous section (Eq. \eqref{equ-success probability model 1}) and with simulation results, for different user populations. Although the simple approximation is slightly less accurate for a very small number of users (e.g., around 5\% for 4 users) it coincides with the results of the simulation \emph{and the approximate model from the previous section} (with dependent queues, which cannot be calculated for large $K$) even for moderate number of users (e.g., around 2\% for 10 users). It is important to emphasize that since the number of users is relatively small, we used the equation in its explicit form, meaning without taking $K$ to infinity. In Figure \ref{fig-Service_time_model2}, we compare the average service time computed according to the $M/M/1$ queue approximation, with service rate $p_{succ}\uptau$, which was calculated according to equation \eqref{equ-Transmission success probability satisfies the equation - Memoryless arrival}, with simulation of the system, which, of course, included dependent queues and variable $p_{succ}$.  Clearly, the approximation shows excellent agreement with the simulation results.

  \begin{figure}[!t]
        \centering
        \begin{subfigure}[b]{0.45\textwidth}
                \includegraphics[width=\textwidth]{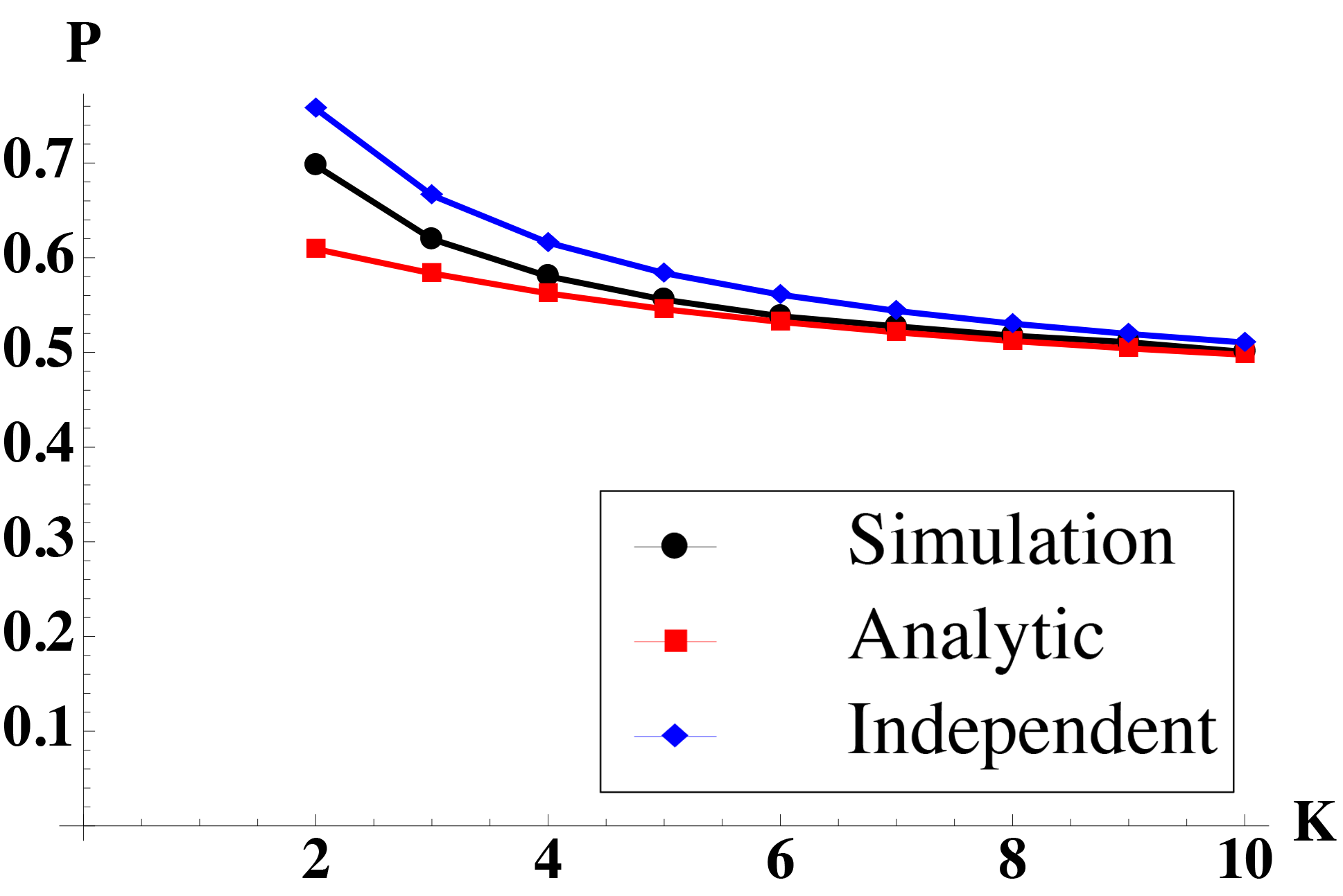}
                \caption[]{}
                \label{fig-SuccessProbability_2-10_0366_independent}
        \end{subfigure}%
        \quad
        \begin{subfigure}[b]{0.45\textwidth}
            \centering
            \includegraphics[width=\textwidth]{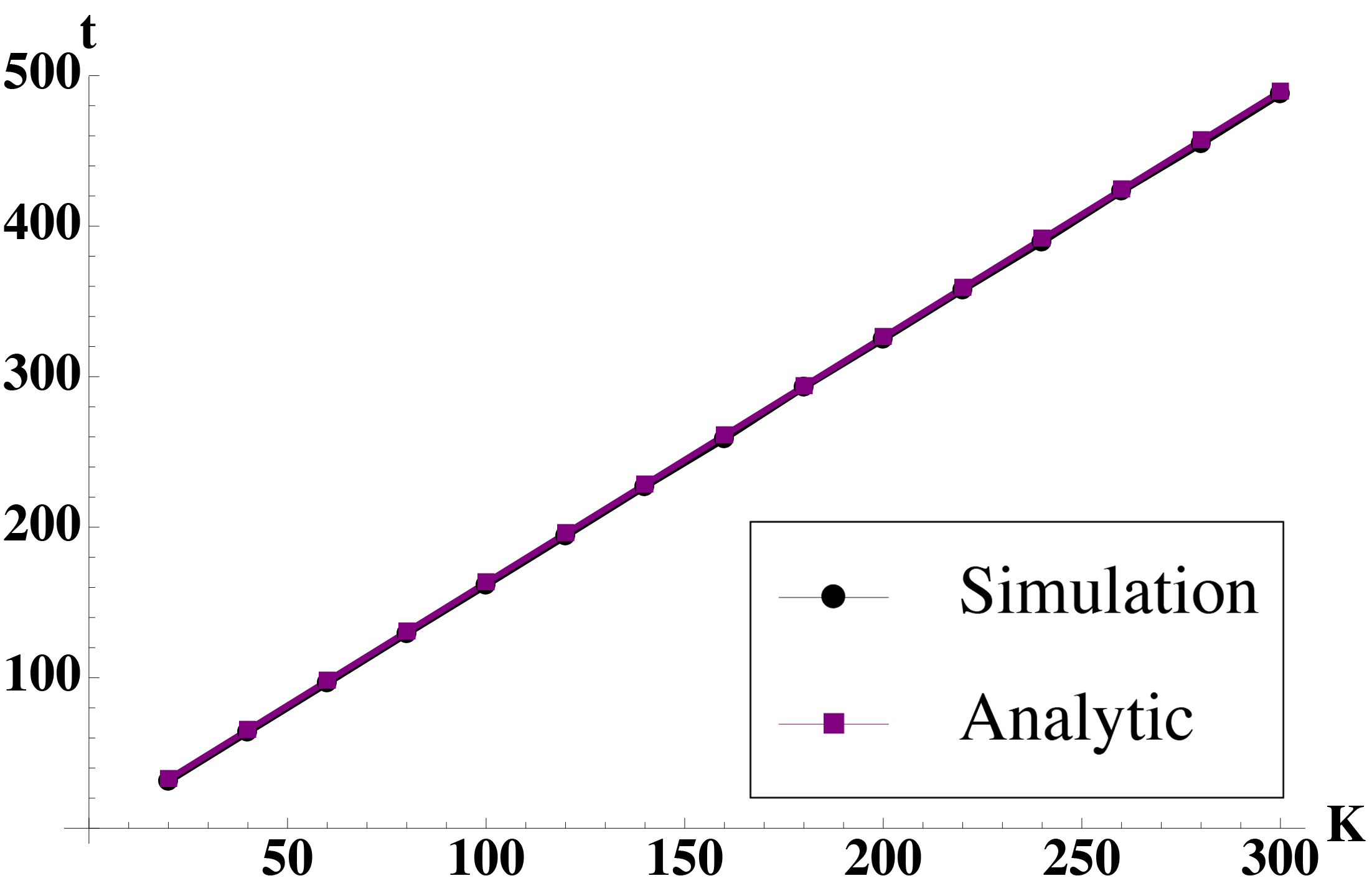}
            \caption[]{}
            \label{fig-Service_time_model2}
        \end{subfigure}
        \caption{Simulation results for the approximate model using constant collision probability. (a) Comparison of $p_{succ}$ between the analytic derivation of the approximation done in Section \ref{Approximate model 1} (equation \eqref{equ-success probability model 1}), simulation and the approximation of $K$ \emph{independent} M/M/1 queues given in Lemma \ref{lem-Transmission success probability satisfies the equation - Memoryless arrival} where the total arrival rate is $\lambda_T=\frac{1}{e}(1-0.001)$. (b) The service time of an M/M/1 queue with service rate $p_{succ}\uptau$, calculated using \eqref{equ-Transmission success probability satisfies the equation - Memoryless arrival}, compared to simulation results of system with $K$ \emph{interdependent} queues where the total arrival rate is $\lambda_T=0.3$.}
    \end{figure}
    
       
   \subsubsection[Threshold value]{Threshold value}\label{Threshold value discussion}
   
   The threshold value has a profound effect on the performance. Different threshold values will support different arrival rates, i.e., a different stability region. In addition, the threshold value affects the users' delay. On the one hand, a high threshold will result in low exceedance probability hence long intervals between transmission attempts. On the other hand, low threshold results in high exceedance probability, hence high collision probability. The optimal threshold value depends on the number of backlogged users in the system at any given time, e.g., when there are only a few backlogged users, a low threshold should be chosen and when many users are backlogged a high threshold should be chosen. Note, however, that the process of monitoring the users at all times, and notifying them regarding the current threshold before each transmission, is not only complex analytically, but mainly impractical in real systems. Therefore, in the analitical part of this work, we focused on a fixed threshold, independent of the users' status. Specifically, we chose a threshold value such that the probability of exceedance is $1/K$. This threshold is conservative, as it is designed for the case that all users are backlogged at all times. However, it is interesting to see how a different fixed value affects the results.

Hence, we examine the effect of the threshold on the performance based on simulation results and specifically, we show that a better threshold value can be chosen (which is less conservative). Figures \ref{fig-System_throughput_as_function_of_threshold} and \ref{fig-System_delay_as_function_of_threshold} depict the throughput and system delay as a function of the exceedance probability, respectively, for 50 users with total arrival rate of $\lambda_T=0.35$. The blue dotted line in Figure \ref{fig-System_throughput_as_function_of_threshold} is the average number of backlogged users as a function of the exceedance probability. Expectedly, when the exceedance probability is high, the throughput decreases and the delay is rapidly growing due to high collision probability. On the other hand, when it is low, throughput decreases as well, due to the long intervals in which no users will attempt transmission, as the threshold is high. This, of course, indicates on the instability of the system. However, as can be seen from figures \ref{fig-System_throughput_as_function_of_threshold} and \ref{fig-System_delay_as_function_of_threshold}, there is a domain of probability values for which the system achieves its maximum throughput and has a low average delay values. One can see that the value of $1/K$ (black dashed line) and the value for which we have a minimum number of backlogged users are found in this domain. It is clear that although the throughput of the system is at its maximum, for this the domain of probability values, the delay and the average number of backlogged users can be reduced if we increase the exceedance probability. A possible criterion for a better excedance probability value may be a threshold which minimizes the average number of backloged users in the system.
In addition, we wish to point out that figure \ref{fig-System_delay_as_function_of_threshold} may be misleading for values which the system is not stable (for $p<0.012$ and $p>0.03$ approximately), where one would expect to see an infinite delay. This is due to the run time of the simulation. Nevertheless, we still can see the sharp jump in the delay which indicate its general behaviour.

    \ifdouble
\else
    
    \begin{figure*}[tp!]
        \centering
        \begin{subfigure}[b]{0.45\textwidth}
                \centering
                \includegraphics[width=\textwidth]{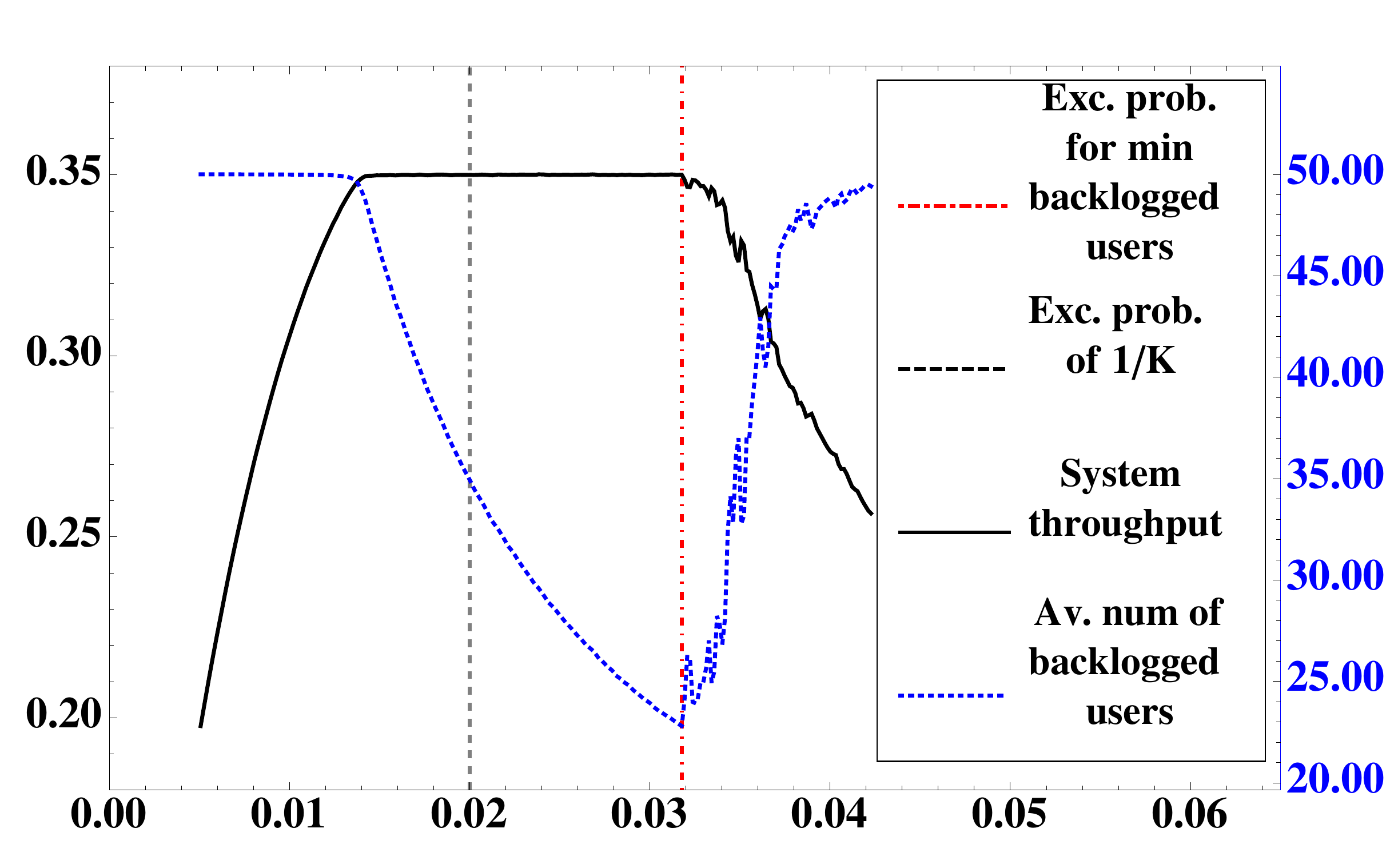}
                \caption{}
                \label{fig-System_throughput_as_function_of_threshold}
        \end{subfigure}%
        \quad
        \begin{subfigure}[b]{0.45\textwidth}
            \centering
            \includegraphics[width=\textwidth]{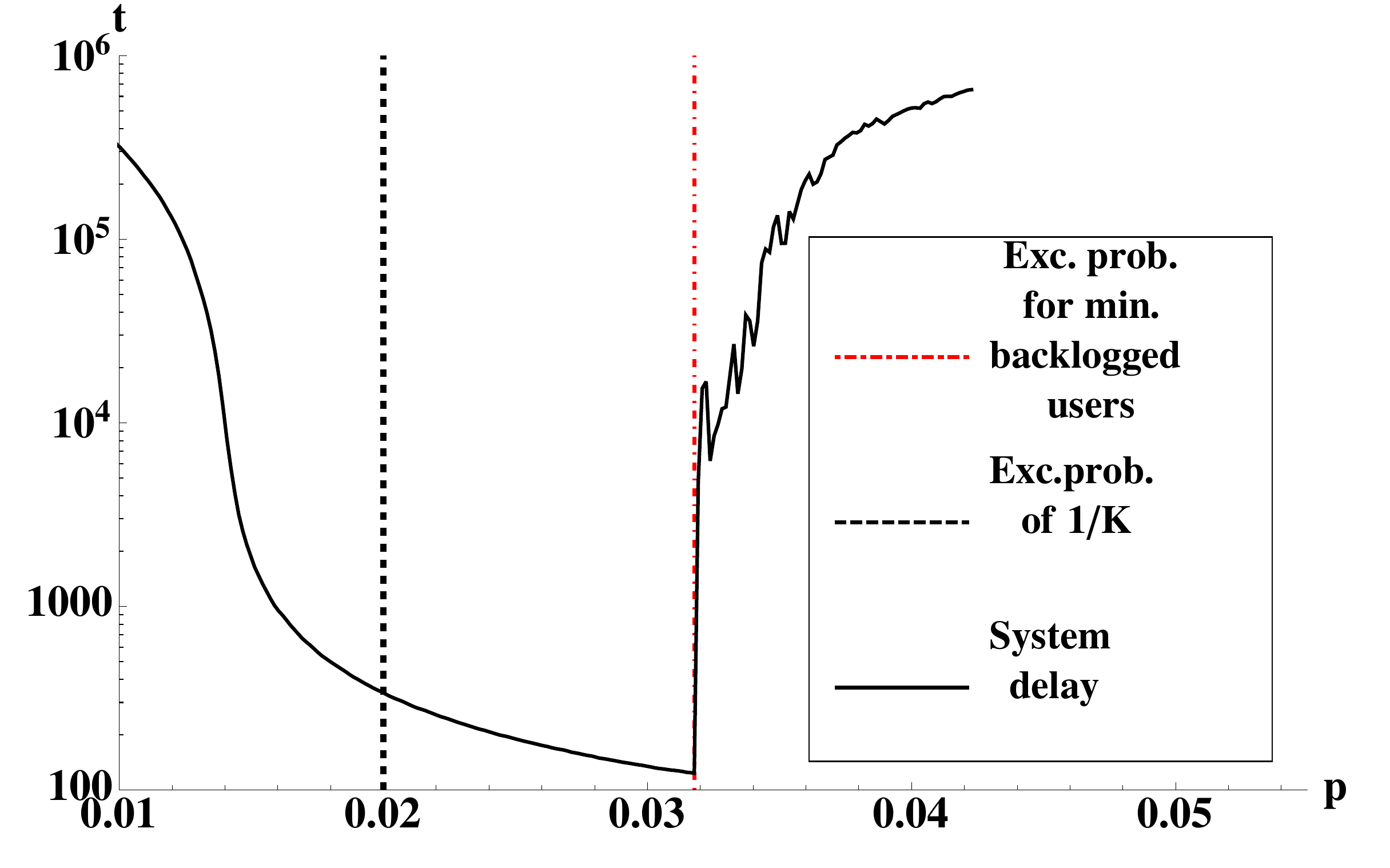}
            \caption{}
            \label{fig-System_delay_as_function_of_threshold}
        \end{subfigure}
        \caption[Throughput and delay as function of threshold]{The system performance metrics as a function of the threshold value. (a) System throughout and average number of backlogged users. (b) System delay. One can compare the performances under the exceedance probability of $1/K$ with different probability values (e.g. for the lower value of the average number of backlogged users). The simulation was preformed for 50 users with total arrival rate of $\lambda_T=0.35$.}
        \label{fig-System_throughput_and_delay_as_function_of_threshold}
    \end{figure*}
    \fi

    \subsection[Queueing Approximate model \Rmnum{3}]{Approximation using Constant Collision Probability - Time Dependent Channel}\label{Approximate model 3}
   
   In many practical scenarios, the channel seen by a user at a given time slot is correlated with the one seen by the user in the previous time slots. Thus, in the sequel, we will analyze the performance of the threshold based algorithm when users are experiencing a time dependent channel distribution. In particular, we assume that the channel capacity distribution experienced by each user is time varying, according to a Gilbert Elliott model \cite{gilbert1960}. Specifically, the channel distribution may be at one of two states, denoted as G (for Good) and B (for Bad), where each state determines a different channel distribution. The transitions between the states Good to Bad and Bad to Good follows a Bernoulli distribution with probabilities $\alpha$ and $\beta$, respectively. Thus, the states evolve according to a 2-state Markov chain as described in Figure \ref{fig-GoodBadchannel}. Note that while the trheshold remains fixed, the threshold exceedance probability clearly depends on the user's channel state. Namely, if a user is in a Good state, the user is expected to exceed the threshold more often than when being in a Bad state. 
Furthermore, the user's collision probability not only depends on the number of other backlogged users but also on the channel state of each such backlogged user. Accordingly, trying to solve the system's stationary distribution, which is one additional dimension over the previous analyzed system, is much more involved than before. We thus adopt the same simplified approach we took in Section~\ref{Approximate model 2}. Specifically, we assume that the collision probability that a user experiences is constant regardless of the number of other backlogged users or their channel states. As before, we assume that the slot duration is small compared to the time interval between two consecutive threshold exceedances, even when the user's channel is in Good state, and approximate the time between threshold exceedances for each user as exponentially distributed with parameter $\mu_g$ and $\mu_b$, depending on whether the user's channel is in Good or Bad state, respectively. Accordingly, the user's service time is exponentially distributed with rates $\mu_g\cdot p_{succ}$ or $\mu_b\cdot p_{succ}$, depending on the user being in Good or Bad state, respectively. Given the Poisson arrival process with rate $\lambda$, the decoupled user's queue model is presented in Figure \ref{fig-Time Dependent Queue}.


     \begin{figure}[!t]
        \centering
        \begin{subfigure}[b]{0.4\textwidth}
          \centering
        \includegraphics[width=0.6\textwidth]{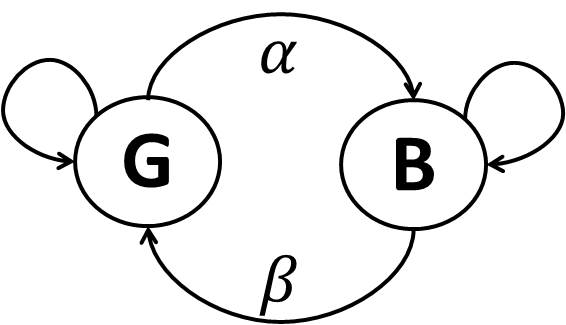}
        \caption{}
        \label{fig-GoodBadchannel}
         \end{subfigure}%
         \begin{subfigure}[b]{0.4\textwidth}
            \centering
            \includegraphics[width=0.8\textwidth]{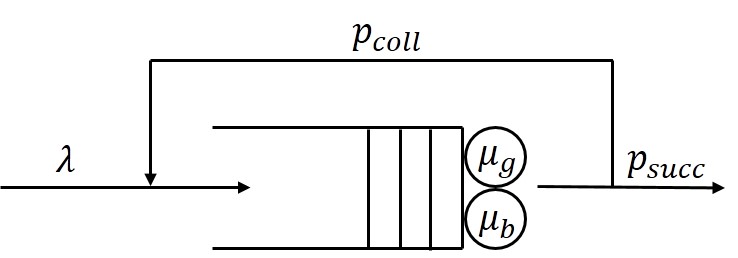}
            \caption{}
            \label{fig-Time Dependent Queue}
        \end{subfigure}%
        
        \begin{subfigure}[b]{0.7\textwidth}
            \centering
            \includegraphics[width=0.9\textwidth]{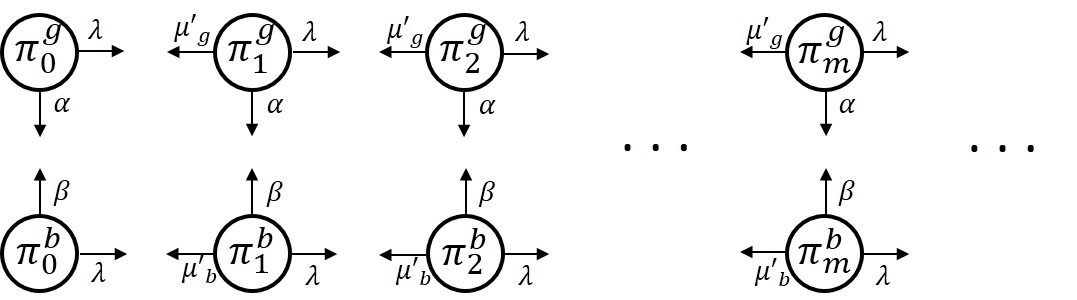}
            \caption{}
            \label{fig-Time Dependent Queue Markov chain}
        \end{subfigure}
        \caption{The Models for the approximation using constant collision probability for time dependent channel. (a) A Good-Bad channel model according to \cite{gilbert1960}. (b) The queue diagram for a user with a time dependent channel and a constant probability for collision. (c) The queue Markov Chain Model for a time dependent user.}
    \end{figure}

        

    This time-dependent queue model can be represented as a continuous Markov process on the set of states $\{\pi^i_m\}$ for $i\in\{b,g\}$, which indicates the Good or Bad state, and $m=0,1,2...$ the number of packets in the queue. This two dimensional Markov chain is presented in Figure \ref{fig-Time Dependent Queue Markov chain}. To ease notation, we denote $\mu_g\cdot p_{succ}=\mu'_g$ and $\mu_b\cdot p_{succ}=\mu'_b$.
    

    In \cite{yechiali1971queuing}, the authors studied a modification of the $M/M/1$ queuing model in which the rate of arrival and the service capacity are subject to Poisson alternations. While the model therin is different than the one here, the analysis suggested relied on a Markov chain which is similar to the one presented in Figure \ref{fig-Time Dependent Queue Markov chain}. The system is solved by using generating-function techniques, resulting in a solution for the steady state probabilities, $\{\pi^i_m\}$, as a function of the transition rate parameters and the root of a third degree polynomial $g(z)$ (only one solution exists under the assumptions). We rely on the solution presented in \cite{yechiali1971queuing}. Next, we resolve the steady state distribution of the chain.
    
    Let $\hat{\mu}$ denotes the average service rate.
    \begin{equation*}
            \hat{\mu}=\pi_g \mu'_g+\pi_b\mu'_b, \quad \text{where,} \quad \pi_g=\frac{\beta}{\alpha+\beta},\quad  \pi_b=\frac{\alpha}{\alpha+\beta}.
    \end{equation*}
  Note that in order to maintain stability, it is required that $\hat{\mu}>\lambda$. We define the partial generating functions of the system as:

    \begin{equation*}
      G_i(z)=\sum^{\infty}_{m=0}\pi^i_m z^m \quad \quad \mid z \mid \leq 1, \ i=g,b.
    \end{equation*}
    From~\cite{yechiali1971queuing}, we have
    \begin{equation*}
      \begin{aligned}
            G_g(z)&=\left( \beta(\hat{\mu}-\lambda)z+\pi^g_0\mu'_g(1-z)(\lambda z-\mu'_b) \right)/g(z) \\
            G_b(z)&=\left( \alpha(\hat{\mu}-\lambda)z+\pi^b_0\mu'_b(1-z)(\lambda z-\mu'_g) \right)/g(z),
        \end{aligned}
    \end{equation*}
    where $g(z)$ is
    \begin{equation}\label{equ-third degree polynomial}
        \begin{aligned}
            g(z)=&\lambda^2z^3-(\alpha\lambda+\beta\lambda+\lambda^2+\lambda\mu'_b+\lambda\mu'_g)z^2  \\
            &+(\alpha\mu'_b+\beta\mu'_g+\mu'_g\mu'_b+\lambda\mu'_b+\lambda\mu'_g)z-\mu'_g\mu'_b.
        \end{aligned}
    \end{equation}
    The steady state probabilities for an empty queue, depending on the channel states, are
    \begin{equation}
            \pi^g_0=\frac{\beta(\hat{\mu}-\lambda)z_0}{\mu'_g(1-z_0)(\mu'_b-\lambda z_0)}, \ \
            \pi^b_0=\frac{\alpha(\hat{\mu}-\lambda)z_0}{\mu'_b(1-z_0)(\mu'_g-\lambda z_0)}
    \end{equation}
    where $z_0$ is the root of the polynomial $g(z)$.
    The remaining steady state probabilities are as follows
    \begin{equation}
        \begin{aligned}
            \pi^g_m&= \pi^g_{m-1}\frac{\lambda}{\mu'_g}+\sum^{m-1}_{j=0}\pi^g_j\frac{\alpha}{\mu'_g}- \sum^{m-1}_{j=0}\pi^b_j\frac{\beta}{\mu'_g} \quad \quad m>0 \\
            \pi^b_m&= \pi^b_{m-1}\frac{\lambda}{\mu'_b}+\sum^{m-1}_{j=0}\pi^b_j\frac{\beta}{\mu'_b} - \sum^{m-1}_{j=0}\pi^g_j\frac{\alpha}{\mu'_b} \quad \quad m>0.
        \end{aligned}
    \end{equation}
    From the first derivative of the partial generating functions, we can attain the expected queue size, which includes the head of line packet. Unlike the analysis in subsection \ref{Approximate model 1}, now

    \begin{equation*}
      \overline{Q}=\frac{\lambda}{\hat{\mu}-\lambda}+\frac{\mu'_g(\mu'_b-\lambda)\pi^g_0+
      \mu'_b(\mu'_g-\lambda)\pi^b_0-(\mu'_g-\lambda)(\mu'_b-\lambda)}{(\alpha+\beta)(\hat{\mu}-\lambda)}.
    \end{equation*}
    Using Little's theorem, one can attain the average waiting time in the queue, $ W=\overline{Q}/\lambda.$
    
    Note that these aforementioned results rely on $\mu'_b$ and $\mu'_g$ which rely on $p_{succ}$ which is assumed to be fixed, yet is unknown and needs to be computed. Hence, in the same manner as in the previous subsection, we define the probability that a specific user attempts transmission, i.e., exceeds the threshold and its queue is not empty, to be
    
   \begin{equation}\label{equ-probability for attempt transmission}
    \begin{aligned}
      P_t &\triangleq Pr(\text{transmission attempt})\\
      &=Pr(\text{exceedance occur \& queue is not empty})\\
      &=(G_g(1)-\pi^g_0)(1-e^{-\mu_g})+(G_b(1)-\pi^b_0)(1-e^{-\mu_b}),
    \end{aligned}
    \end{equation}
    where the first term consists of the probability to be in a Good state with a packet to transmit, times the probability which an exceedance occurs while existing in the Good state; The second term is similar, referring only to the Bad state.
The probability for success, given one user is about to transmit, can be obtained by a similar calculation as in Lemma \ref{lem-Transmission success probability satisfies the equation - Memoryless arrival}, and we have
    \begin{equation}\label{equ-probability for sucss-third model}
      p_{succ}=(1-P_t)^{K-1}.
    \end{equation}
    
    Since the root $z_0$ and $p_{succ}$ are coupled, \eqref{equ-probability for sucss-third model} and the roots of the third degree polynomial \eqref{equ-third degree polynomial} must be solved simultaneously in order to compute them both. 
    
       

    In order to evaluate the approximation, we ran a set of simulation. In the results presented here, the transition rates between Good and Bad states are set equal. Specifically, $\alpha=\beta=0.1$. The users were considered homogenous. The arrival and service rate parameters, presented in the figures, depict the total rates of the system, and were divided equally among all the users. The service rates were set to $\mu_g=0.7$ and $\mu_b=0.5$ for Good and Bad channel quality, respectively. 
Figure \ref{fig-QueuingPerformance_time dependent_P_succ} depicts the probability for success, $p_{succ}$ vs.\ the number of users for two different arrival rates, $\lambda_T=0.1$ and $\lambda_T=0.3$. The figure clearly depicts that the estimated probability for success $p_{succ}$, coincides with the simulation results for both arrival rates. Figure~\ref{fig-QueuingPerformance_time dependent} depicts the mean queue size and the mean sojourn time vs.\ the number of users for $\lambda_T=0.3$, which also shows good agreement between the approximated analytical results and the simulation results.

\ifdouble
\else
	\begin{figure*}[t!]
        \centering
        \begin{subfigure}[b]{0.45\textwidth}
        \center{
                \includegraphics[width=\textwidth,height=4cm,keepaspectratio]{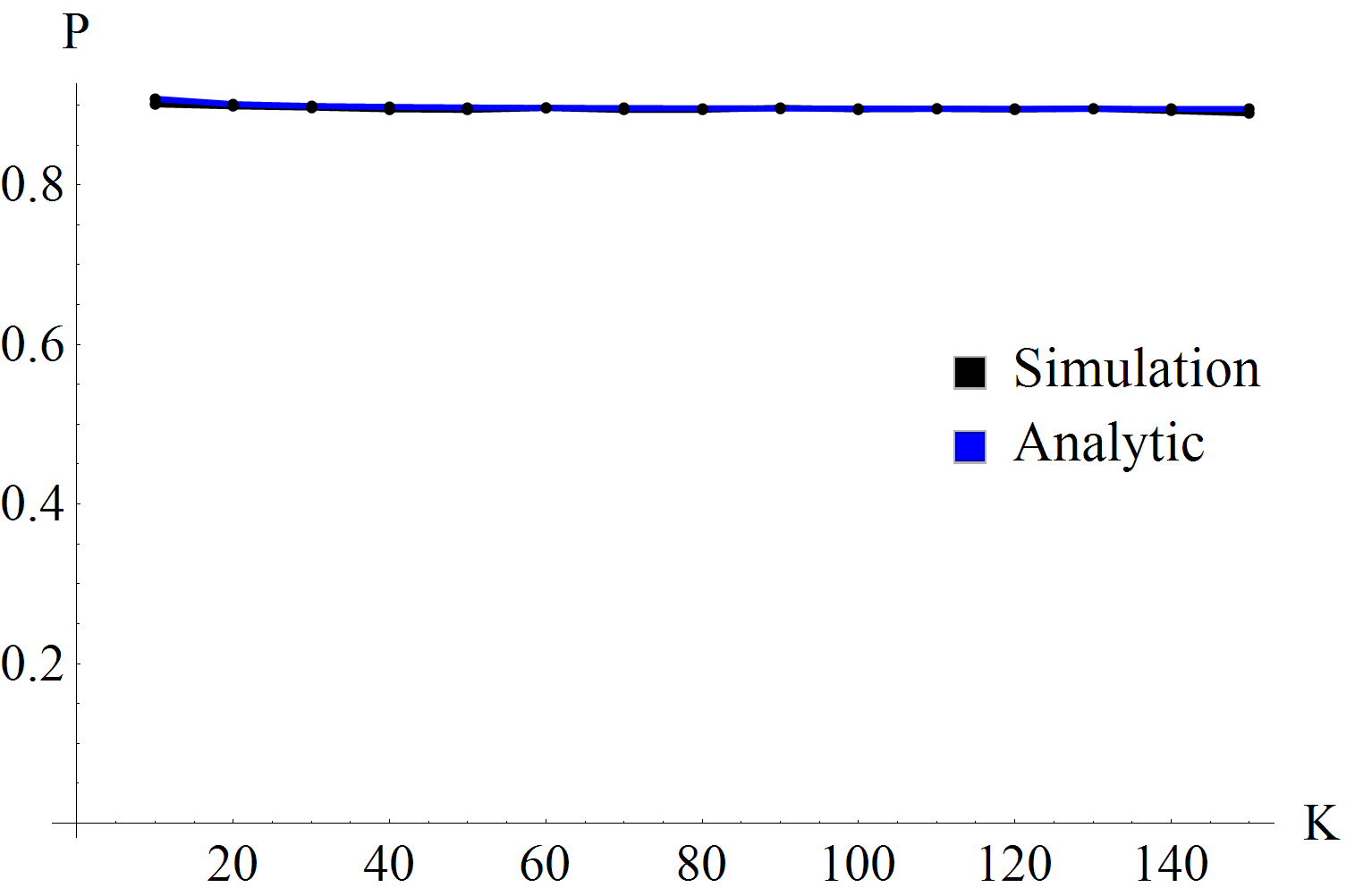}
                \caption{}}
                \label{fig-P_succ_L=0.1}
        \end{subfigure}%
        \begin{subfigure}[b]{0.45\textwidth}
                \center{
                \includegraphics[width=\textwidth,height=4cm,keepaspectratio]{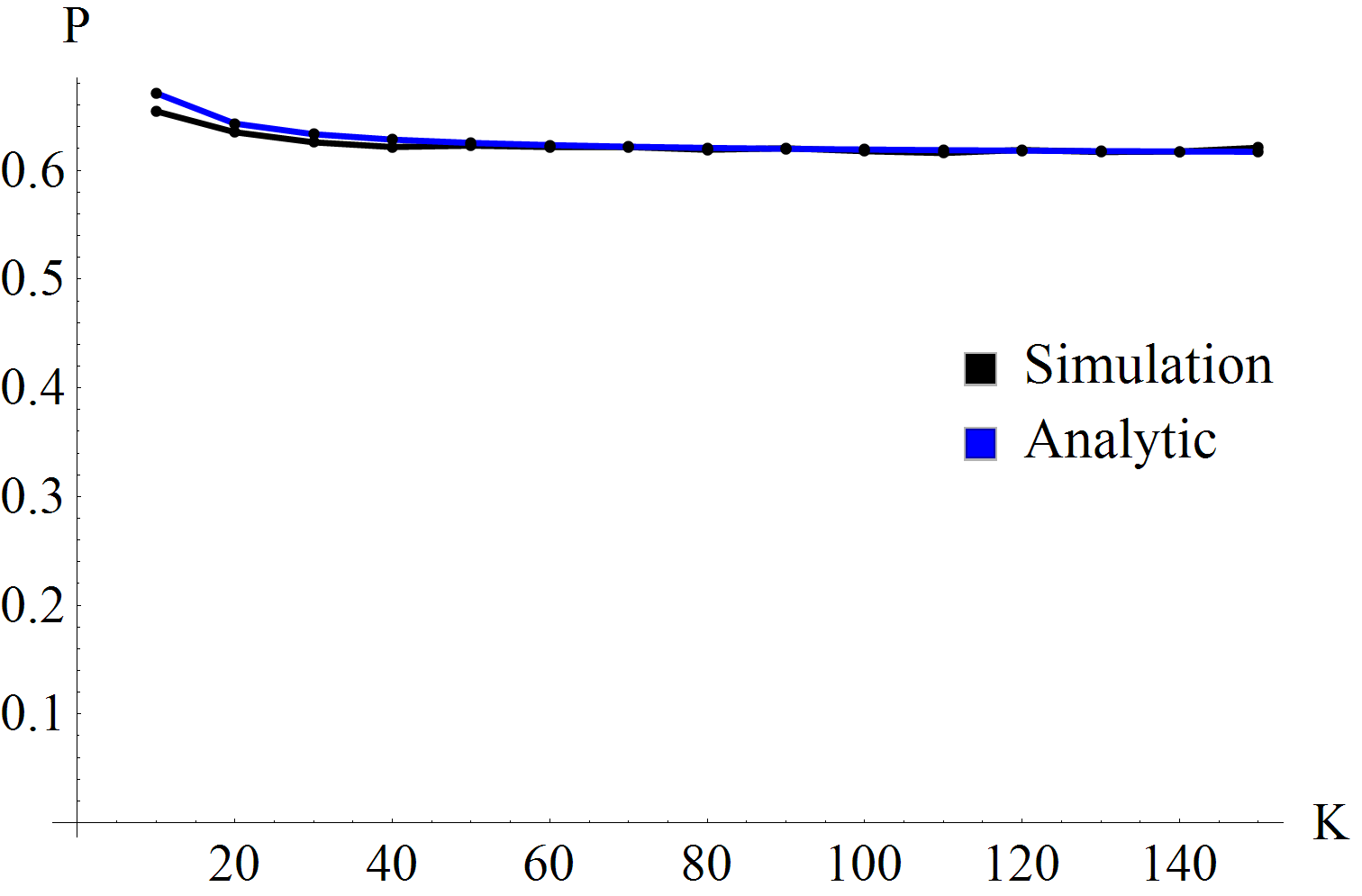}
                \caption{}}
                \label{fig-P_succ_L=0.3}
        \end{subfigure}
        \caption[Success probability time dependent model]{Success probability for the approximation using constant collision probability for time dependent channel, as given in \eqref{equ-probability for sucss-third model}, compared to simulation results of a time dependent queueing system, according to the Good-Bad channel model, as a function of the number of users. Where $\mu_g=0.7,\mu_b=0.5$ for (a)$\lambda_T=0.1$ (b)$\lambda_T=0.3$.}
        \label{fig-QueuingPerformance_time dependent_P_succ}
    \end{figure*}

    \begin{figure*}[!t]
        \centering
        \begin{subfigure}[b]{0.45\textwidth}
                \centering
                \includegraphics[width=\textwidth]{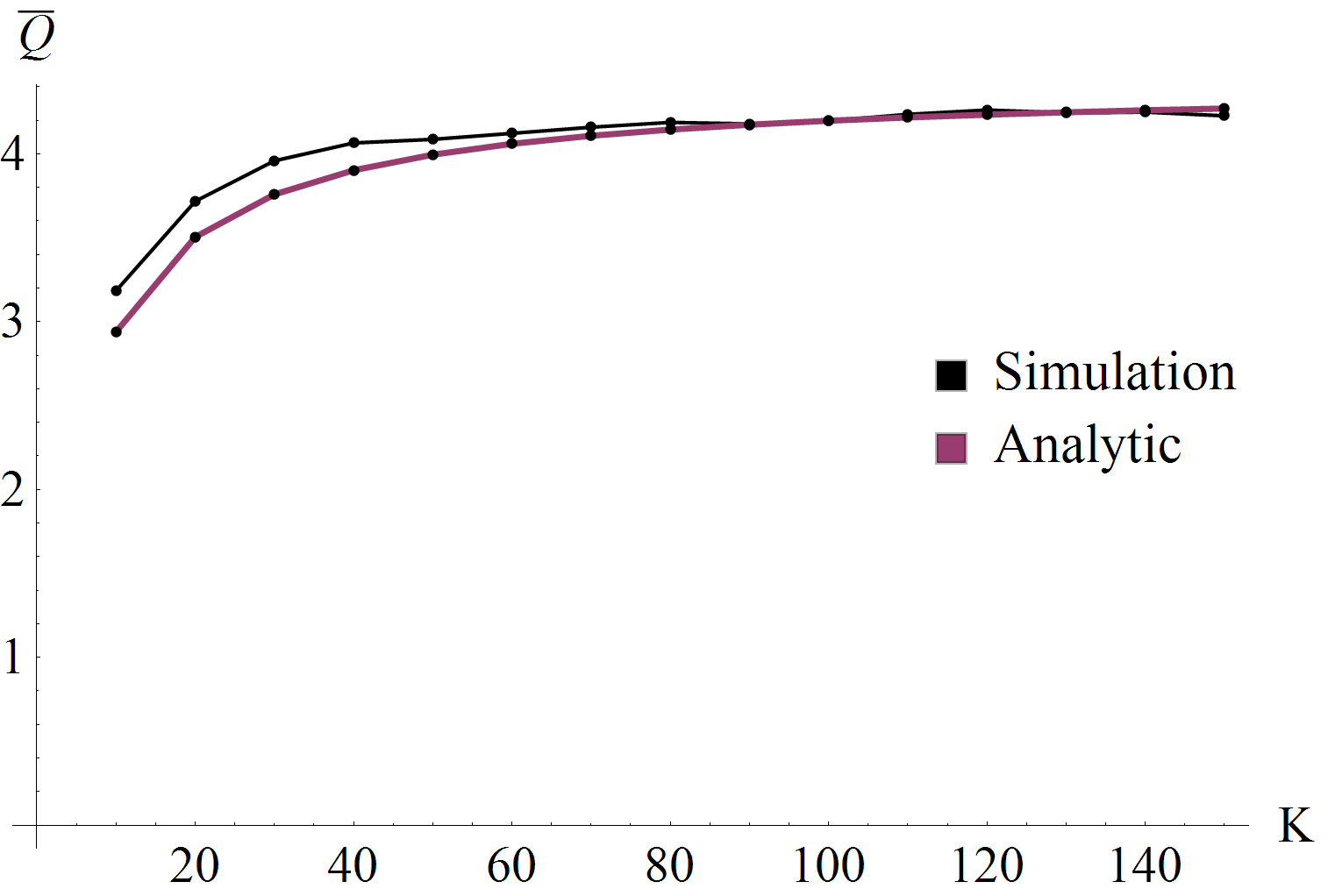}
                \caption{}
                \label{fig-MeanQueueSize_L=0.3}
        \end{subfigure}
        \begin{subfigure}[b]{0.45\textwidth}
                \includegraphics[width=\textwidth]{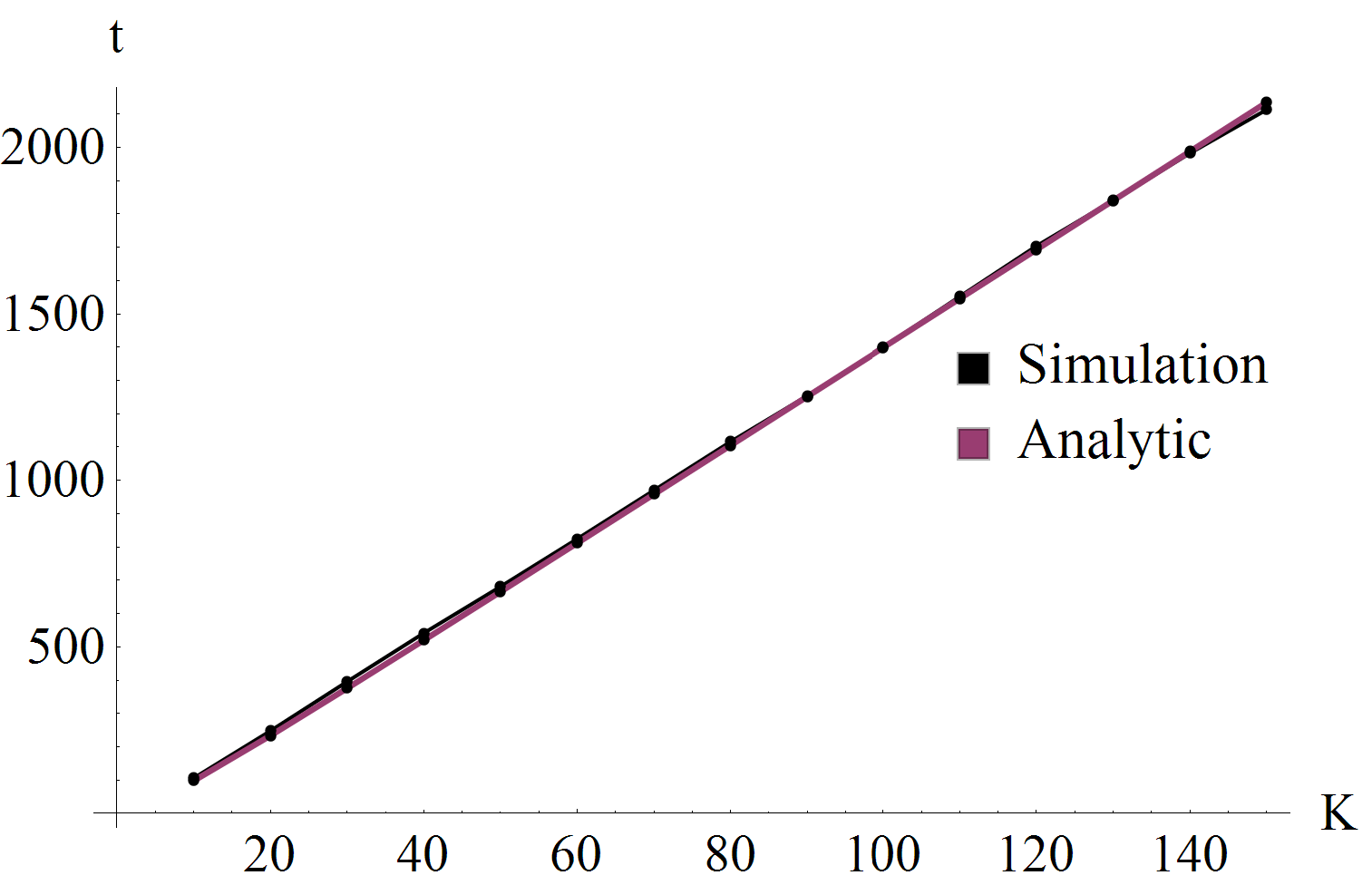}
                \caption{}
                \label{fig-TimeInLine_L=0.3}
        \end{subfigure}
        \caption[System performance time dependent model]{Performance metrics according to the analytic derivation of the approximation using constant collision probability for time dependent channel compared to simulation results of a time dependent queueing system, according to the Good-Bad channel model, as a function of the number of users. Where $\mu_g=0.7,\mu_b=0.5$ and $\lambda_T=0.3$. Where (a) mean queue size and (b) time in queue.}
        \label{fig-QueuingPerformance_time dependent}
    \end{figure*}
\fi

\section{Scaling Law Under Time Dependent Channel}\label{sec-Scaling Law Under Time Dependent Channel}
    
    Up until now, we explored the performance of the MAC system for time independent and time dependent channel scenarios. We considered a distributed threshold based scheduling algorithm, which ensures that when an exceedance occurs, the transmitting user transmits with high channel capacity. This is due to the high threshold value which is set such that only one user will exceed it on average. As mentioned, this scheme exploit multi-user diversity, i.e., let the best user utilize the channel. The scaling laws for independent channels were studied in \cite{qin2003exploiting,qin2006distributed,kampeas2014capacity}. However, scaling law of the channel capacity for time dependent channels, to the best of our knowledge, was not considered yet.  Hence, in this section, we derive the scaling laws of the channel capacity for our MAC system under the Good-Bad channel model described earlier. We note here that, for this analysis, we consider the users to be backlogged at all time. 
    
    We start by formulating the problem and analysing the scaling laws under a centrelized scheduling algorithm, where the base station chooses the strongest user for transmission using the channel state information (CSI) sent to it from the users. We then return to for the distributed algorithm and analyse it as well.    \\

    In order to exploit user diversity, the user with the best channel capacity must utilize the channel. Therefore, the problem is in finding the distribution of the random variable $\widetilde{M_K}$, the maximum capacity in a time dependent channel,  defined as:
    \begin{equation}\label{equ-CapacityDefinition}
            \widetilde{M_K}=\max\{C_1(n),C_2(n),...,C_K(n)\},
    \end{equation}
    where the capacity $C_i(n)$ of the $i$-th user in each slot is determined by the Good-Bad Markov process $\{J(n)\}$:
    \begin{equation}\label{equ-UserCapacityProcess}
    C_i(n)=
            \begin{cases}
                N_g    & \text{when  } J_i(n)=Good \\
                N_b    & \text{when  } J_i(n)=Bad.
            \end{cases}
    \end{equation}
    
$N_g$ and $N_b$ are random variables distributed normally with parameters $(\mu_g,\sigma_g)$ and $(\mu_b,\sigma_b)$, respectively. That is, we assume that in each slot, each user exists in either a Good or a Bad state, distinguished by different parameters of the Gaussian distribution which models the capacity. This is due to the Gaussian approximation for the MIMO\footnote{The assumption of a MIMO channel is used only to have a concrete expression for the capacity, with a reasonable approximation (in this case, as Gaussian random variable; see e.g., \cite{shmuel2014capacity}).} channel capacity \cite{smith2002gaussian,chiani2003capacity}. The parameters reflect the differences in the channel qualities, in a way that a good channel parameters maintain $(1) \ \sigma_g>\sigma_b, \ \mu_g,\mu_b\in\mathds{R}$ or $(2) \ \sigma_g=\sigma_b$ and $\mu_g>\mu_b$.\\

    We are interested in the expected channel capacity $E[\widetilde{M_K}]$ at the limit of large $K$. Here, we first use results from EVT, as well as results concerning time dependent processes in order to evaluate the limit distribution of the maximal value when we consider centralized scheduling. Then, when considering distributed scheduling, we use results  from PPA  in order to analyse threshold arrival rates and tail distributions.

    \subsection{Centralized Scheduling}
    Let us consider the distribution of each $C_i(n)$, as defined in \eqref{equ-UserCapacityProcess}, which is determined by the Good-Bad Markov process $\{J(n)\}$. The stationary distribution of the chain is $(\frac{\beta}{\alpha+\beta},\frac{\alpha}{\alpha+\beta})$, which we will denote in short as $(p,q)$. We have
    
    \ifdouble
\else
   \begin{equation}\label{equ-stationary distribition of C_i(n)}
    \begin{aligned}
                F(x)\eqdef &P(C_i(n)\leq x)= P(N_g(n)\leq x \mid J(n)=G )\\
                 & \qquad \cdot P(J(n)=G)+P(N_b(n)\leq x \mid J(n)=B )P(J(n)=B)\\
                = &pF_g(x)+qF_b(x),
     \end{aligned}
    \end{equation}
    \fi
    where $F_g(x)$ and $F_b(x)$ are Gaussian distributions with parameters $(\mu_g,\sigma_g)$ and $(\mu_b,\sigma_b)$, respectively. The distribution of the maximal value is in the form of
    \begin{equation*}
            P(\widetilde{M_K}\leq x)= P(C_1(n)\leq x,...,C_K(n)\leq x)=F^K(x),
    \end{equation*}
    due to the independence \emph{between the users}. We wish to test the behaviour of $F^K(x)$ as $K \rightarrow \infty$.
	Note that this means we actually use the stationary distribution as the marginal distribution. 

    
    Our main result in this context is the following.
    \begin{theorem}\label{thm-Capacity distribution in a time dependent channel convergence to a Gumbel}
            Let $(C_1,...,C_K)$ be the sequence of the users' capacities in a certain time slot where each capacity has a distribution $F(x)$ as given in \eqref{equ-stationary distribition of C_i(n)}. Then, the asymptotic distribution of the scheduled user's capacity, i.e., $\widetilde{M_K}=\max\{C_1,...,C_K\}$, is a Gumbel distribution. Specifically,
            \begin{equation}\label{equ-Capacity distribution of theorem}
                    P\{a_K(\widetilde{M_K}-b_K)\leq x\}\rightarrow e^{-e^{-x}}
            \end{equation}
            where
            
            \ifdouble
\else
            \begin{equation}
                    a_K= \frac{\sqrt{2\log{K}}}{\sigma_g}, \ \ b_K= \sigma_g\left(\sqrt{2\log K}-\frac{\log{\log K}+\log{\frac{4\pi}{p^2}}}{2\sqrt{2\log K}}\right)+\mu_g.
            \end{equation}
            \fi
            Therefore, the expected throughput of the transmitting user is 
            \ifdouble
\else
    \begin{equation}\label{equ-Expected channel capacity (stationary distribution)}
                E[\widetilde{M_K}]= b_K+\frac{\gamma}{a_K}=\sigma_g\left(\left(\sqrt{2\log K}-\frac{\log{\log K}+\log{\frac{4\pi}{p^2}}}{2\sqrt{2\log K}}\right)+\frac{\gamma}{\sqrt{2\log{K}}}\right) +\mu_g,
    \end{equation}
    \fi
    where $\gamma=0.57721$ is the Euler-Mascheroni constant.
    \end{theorem}
    
    The proof of Theorem \ref{thm-Capacity distribution in a time dependent channel convergence to a Gumbel} is based on results we obtained on EVT for stationary processes. The complete proof of Theorem \ref{thm-Capacity distribution in a time dependent channel convergence to a Gumbel} is given in Appendix \ref{Appendix C}.
    

\begin{figure}[t]
                \centering
                \includegraphics[width=0.6\textwidth]{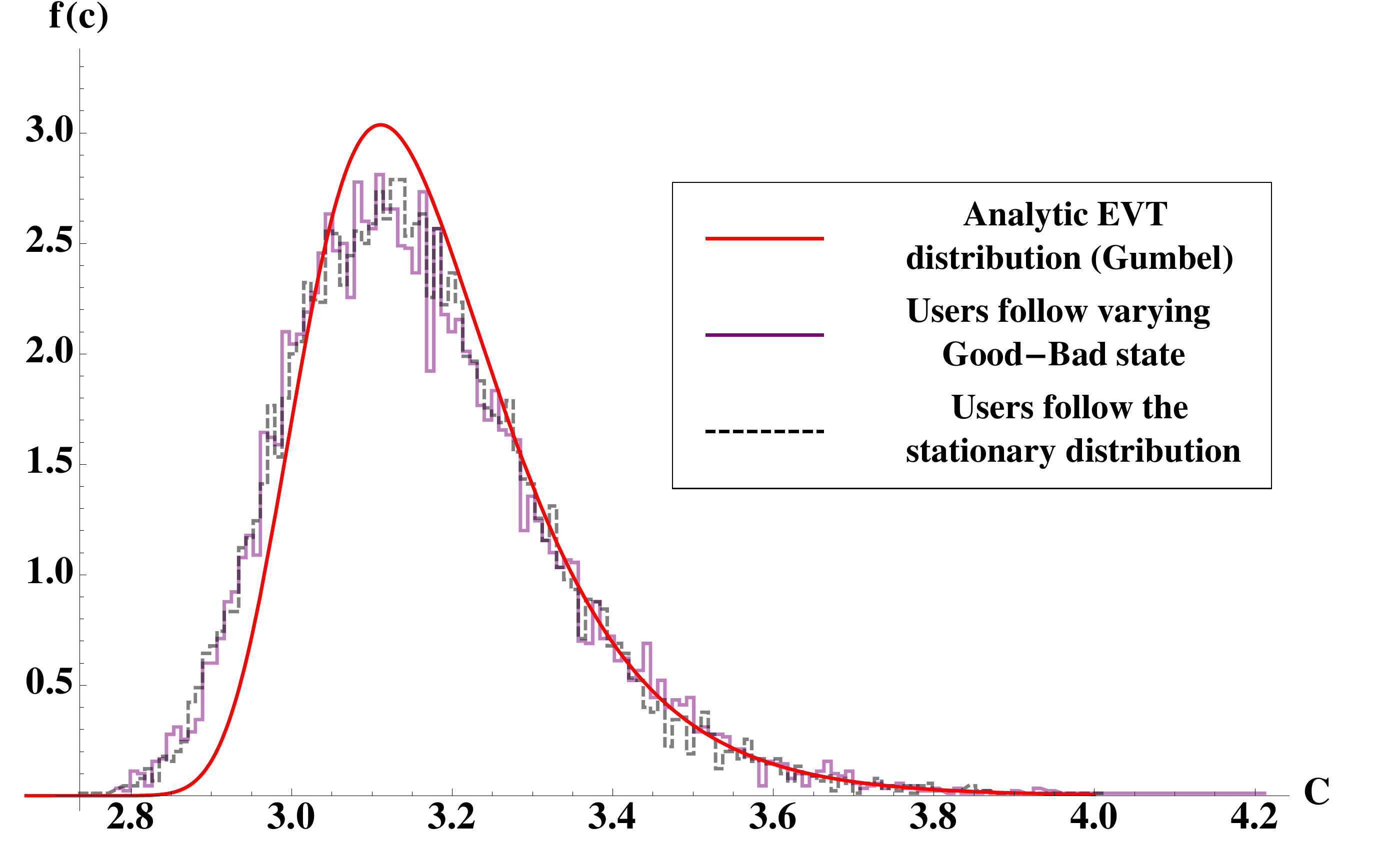}          
        \caption[Simulation for capacity distribution]{The maximal capacity distribution where the capacities were drawn according to the stationary distribution (dashed line), according to the a varying Good-Bad process for each user (solid line) compared to the corresponding Gumbel density with the constants $a_K$ and $b_K$ (red line). For 5000 users.}
        \label{fig-MaxialCapacityDistributionForTimeDependentChannel_5000}
\end{figure}

    Note that the proof of Theorem \ref{thm-Capacity distribution in a time dependent channel convergence to a Gumbel} is based on the assumption that $p,q \neq 0$ which implicitly implies that $p$ and $q$ are bounded away from zero. Nevertheless, note that the expected throughput of the transmitting user when setting $p=1$, which implies that all users have a Good channel state, agrees with the one presented in \cite{kampeas2014capacity}. Simulation results for the capacity distribution are given in Figure \ref{fig-MaxialCapacityDistributionForTimeDependentChannel_5000}. The figure clearly depicts a match between the analysis and simulation results.

    The probability $p$, which is the stationary probability to exist in a Good state, governs the average number of users which are in Good state. That is, in each time slot, one can distinguish between two groups of users: the users that are in a Good state and the users that are in Bad state.  It is easy to show that the number of good users, on average, is $pK$. Hence, as $p$ grows, the expected capacity grows since there are more users in a Good state. This is clear from analytical result in \eqref{equ-Expected channel capacity (stationary distribution)} for the channel capacity as well as in Figure \ref{fig-Capacity comparison as Function of p}.

    One may wonder why considering the bad group at all, meaning, why should a user which is in a Bad channel state be taken into account in the scheduling decision process. Leaving only the users with the good channel to compete for the channel, the capacity with $K$ sufficiently large is
    \begin{equation}\label{equ-CapacityExpressionOnlyGood}
        E[\widetilde{M_{pK}}] = \sigma_g\sqrt{(2\log pK)}+\mu_g+\text{o}\left(\frac{1}{\sqrt{\log{pK}}}\right).
    \end{equation}

    Figure \ref{fig-Capacity comparison as Function of p} depicts the influence of the bad group on the capacity. For rather small values of $p$, it is beneficial to schedule the strongest user \emph{from both groups}. As $p$ grows, the size of the good group grows as well, hence the two curves converge to the case were all the users are in Good state. Thus, for small values of $p$ the users in the \emph{Bad} state still have a significant impact, in the sense that they have a high enough probability to be the strongest and gain channel access. Figures \ref{fig-capacityGainFforTimeDependentChannel_p=0.5} and \ref{fig-capacityGainFforTimeDependentChannel_p=0.2} give a different perspective: the capacity as a function of the number of users.  In Figure \ref{fig-capacityGainFforTimeDependentChannel_p=0.2}, the difference between the capacities is noticeable.

 \begin{figure}[t]
        \centering
        \begin{subfigure}[b]{0.33\textwidth}
                \centering
                \includegraphics[width=\textwidth,height=4cm,keepaspectratio]{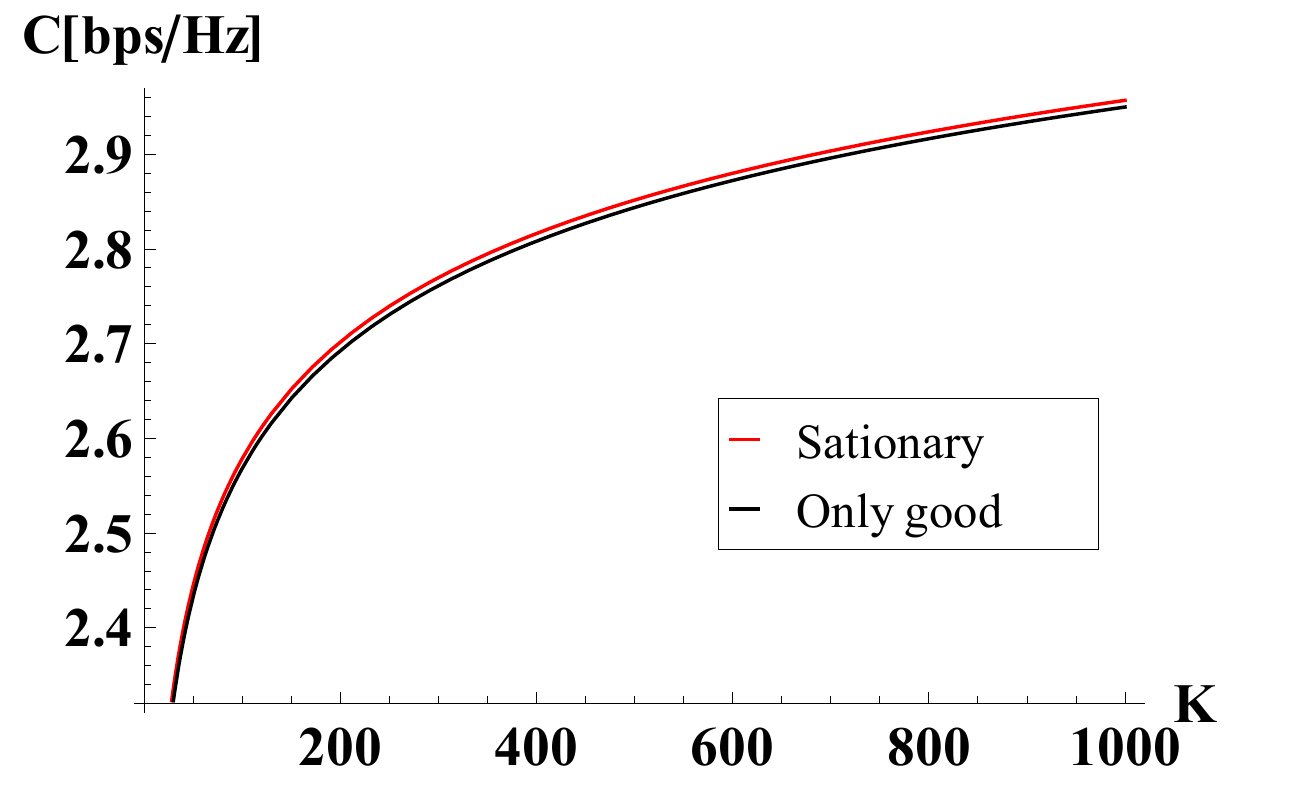}
                \caption{}
                \label{fig-capacityGainFforTimeDependentChannel_p=0.5}
        \end{subfigure}%
        \begin{subfigure}[b]{0.33\textwidth}
                \includegraphics[width=\textwidth,height=4cm,keepaspectratio]{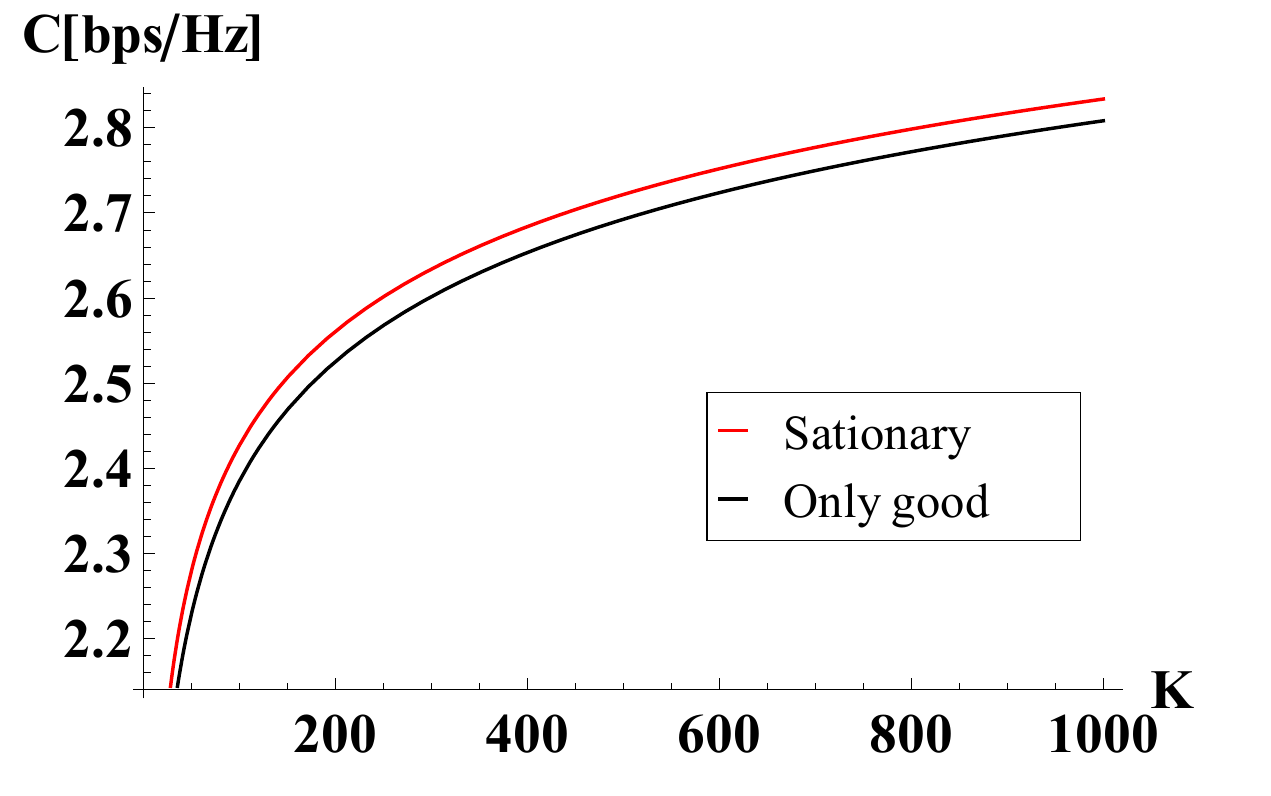}
                \caption{}
                \label{fig-capacityGainFforTimeDependentChannel_p=0.2}
        \end{subfigure}%
        \begin{subfigure}[b]{0.33\textwidth}
                \includegraphics[width=\textwidth,height=4cm,keepaspectratio]{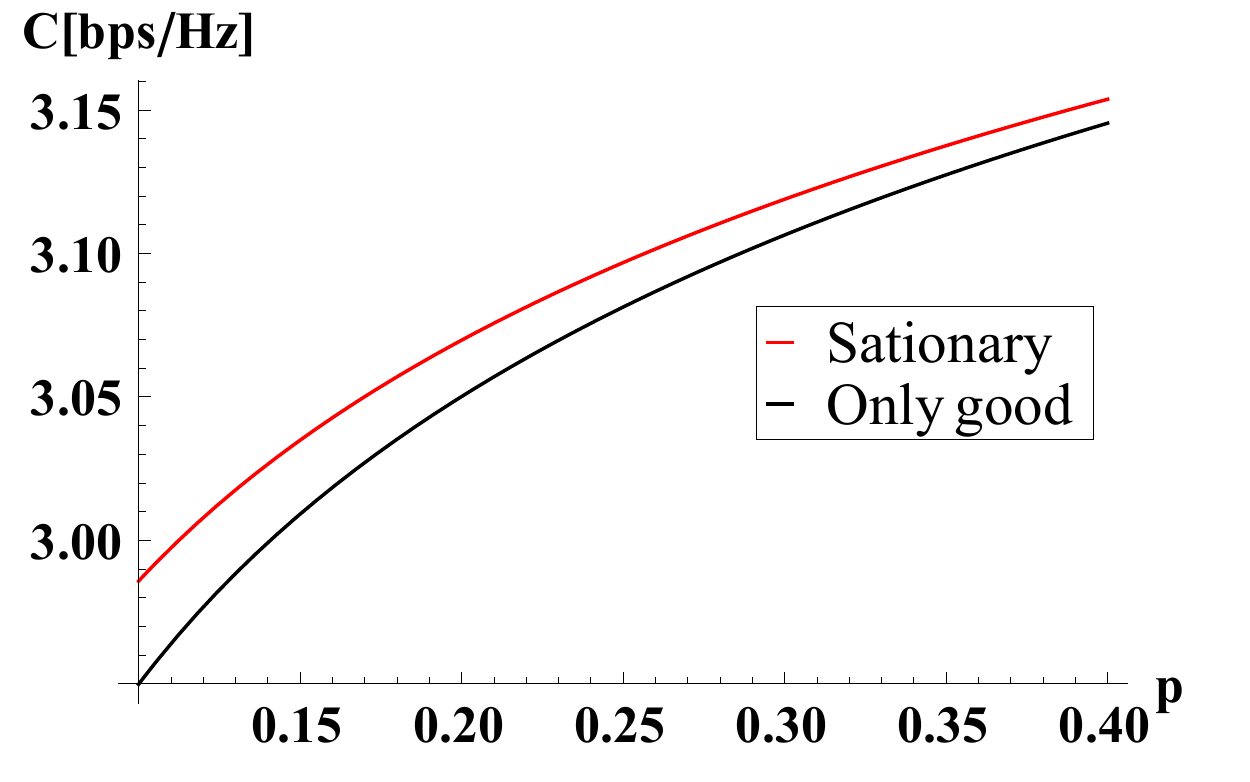}
                \caption{}
                \label{fig-Capacity comparison as Function of p}
        \end{subfigure}
        \caption[Capacity gain comparison]{Capacity comparison for choosing only from the "good" group, of size $Kp$, and the whole population (as given in \eqref{equ-Expected channel capacity (stationary distribution)}) as a function of the number of users with fixed $p$, where in (a) p=0.5 and in (b) p=0.2. And as a function of $p$, for $K=5000$ users in (c). Here $\mu_g = \sqrt{2}$ and $\sigma_g = 0.5$.}
        \label{fig-gain for time dependent channel}
    \end{figure}

\subsection{Distributed Scheduling}
	
	One of the main disadvantages in centralized scheduling is the overhead (e.g., CSI) which is needed for communicating properly. Of course for multi-user systems this drawback is more acute. Hence, distributive scheduling schemes become more attractive. One distributive approach  which has shown excellent exploitation of multi-user diversity, lets the users transmit based on their instantaneous channel condition \cite{qin2006distributed}. Specifically, given a predefined threshold, users may attempt transmission only if their channel gain exceeds it.
	When considering this distributed scheduling algorithm, we derive the expected channel capacity using PPA. This method can also be found in \cite{kampeas2014capacity} for the case of \textit{i.i.d.} and heterogeneous users. We first give a brief review on the construction of the point process and the relevant results which will be used throughout this subsection.
	
	Let $\{X_n\}$ be a sequence of standard normal \textit{i.i.d.} variables, with the Gumbel Distribution as the extreme value distribution and normalizing constants $a_n$ and $b_n$. We then define a sequence of points on $[0,1]\times \mathds{R}^2$ by:
      \begin{equation}\label{equ-sequence of point processes N_n}
          N_n=\Big\{\frac{i}{n},X_i:i=1,...,n \Big \}.
      \end{equation}
	Under the above, we have the following Theorem.
    \begin{theorem}\label{thm-Point process convergance of iid}(\cite{EVT:Springer1983})
        The sequence $N_n$ on $[0,1]\times (u,\infty)$, for some large value of $u$, converge to a non-homogeneous Poisson process $N$ with parameter $\uptau$, that is, $N_n\rightarrow N \ \ \text{as } \ n\rightarrow \infty$.
      \end{theorem}    

	However, since we are interested in time dependent environment, we would like to explore the point process analysis for more general stationary sequences. In \cite{leadbetter1976weak}, exceedances of a high level $u_n$ by a stationary sequence $\{X_i\}$ (i.e., points where $X_i>u_n$), were analyzed, obtaining Poisson limits under weak dependence restrictions. In particularly, the following conditions make precise the notion of extreme events being near-independent if they are sufficiently distant in time.
	\begin{Definition}\label{def-StrongMixing}
        We say $\{X_n\}$ is strongly mixing if there is a function $g$ on the positive integers with $g(k) \rightarrow 0$ as $k \rightarrow \infty$ such that, if $A \in \mathfrak{F}(X_1,...,X_m)$ and $B \in \mathfrak{F}(X_{m+k},X_{m+K+1},...)$ for some $k, m\geq 1$, then $$\mid P(A\cap B)-P(A)P(B) \mid \leq g(k),$$ where $\mathfrak{F}(\cdot)$ denotes the $\sigma$-field generated by the indicated random variables.
      \end{Definition}
      When trying to weaken the strong mixing condition, one notes that the events of interest in extreme value theory are typically those of the form $\{X_n\leq u\}$. Hence we have,
      \begin{Definition}\label{def-D condition}(\cite{EVT:Springer1983})
    A stationary series $\{X_n\}$ is said to satisfy the $D(u_n)$ condition if \\ for all $i_1<...<i_p<j_1<...<j_q$ with $j_1-i_p>l$,
    
        \ifdouble
\else
        \begin{multline*}
                |P_r\{X_{i_1}\leq u_n,...,X_{i_p}\leq u_n, X_{j_1}\leq u_n,...,X_{j_q}\leq u_n\} - \\
                P_r\{X_{i_1}\leq u_n,...X_{i_p}\leq u_n\}P_r\{ X_{j_1}\leq u_n,...,X_{j_q}\leq u_n\}|\leq\alpha(n,l),
        \end{multline*}
        \fi
        where $\alpha(n,l)\rightarrow 0$ for some sequence $l_n$ such that $l_n/n \rightarrow 0$ as $n\rightarrow \infty$.
      \end{Definition}
    Another condition which is highly relevant is the local dependence condition $D'(u_n)$:
      \begin{Definition}\label{def-D' condition}(\cite{EVT:Springer1983})
    We say that $D'(u_n)$ is satisfied if, as $k\rightarrow \infty$,
        \begin{equation*}
     \limsup_{n \rightarrow \infty} n\sum_{j=2}^{ \lfloor n/k\rfloor } P\{X_1>u_n,X_j>u_n\}\rightarrow 0.
        \end{equation*}
      \end{Definition}
	
	Considering the above, we have,
      \begin{theorem}\label{thm-Point process convergance of stationary process}(\cite{leadbetter1976weak})
            Let $D(u_n),D'(u_n)$ hold for the stationary sequence $\{X_i\}$ with $u_n=u_n(\uptau)$, such that $n(1-F(u_n))=nP\{X_1>u_n\}\rightarrow \uptau$ as $n\rightarrow\infty$ for all $\uptau > 0$ . Let $N_n$ be the point process, consisting of the exceedances of $u_n(\uptau)$. Then $N_n \overset{d}{\rightarrow} N$ as $n\rightarrow\infty$, where $N$ is a Poisson process with parameter $\uptau$.            
      \end{theorem}
          
	Using Theorem \ref{thm-Point process convergance of stationary process}, we now analyze each user separately and examine its sequence of channel capacities over time. Specifically, we will show that the exceeding points in this \emph{dependent sequence} converge to a Poisson process, similar to a one if the sequence was \textit{i.i.d.}.

    The sequence $\{C_i(n)\}$ depends on which state user $i$ exists in, which is, in turn, governed by the underlying Markov chain process $\{J_i(n)\}$, as defined in \eqref{equ-UserCapacityProcess}. Such processes have been studied before,  \cite{janssen1969processus,denzel1975limit}, and are known as  $\tilde{J}-X$ processes or "chain-dependent" processes.
    In our context, we shall consider the sequence $\{X_n\}$ as the sequence of capacities $\{C_i(n)\}$ of the i-th' user over time, and the chain process $\{\tilde{J}_n\}$ as the sequence of the irreducible, aperiodic, 2-state Good-Bad Markov chain $\{J_i(n)\}$. Since we analyze the capacity process of one user, and all have the same distribution we omit the user index. We have

    \begin{equation}\label{equ-J-C process defenition our model}
        \begin{aligned}
            P(J_n &= j , C_n \leq \alpha \mid J_0,C_1,J_1,...,C_{n-1},J_{n-1}=i) \\
               &= P(J_n=j,C_n\leq \alpha \mid J_{n-1}=i)=P_{ij}H_j(\alpha),
        \end{aligned}
    \end{equation}
    where $P$ is the transition matrix of the chain and $H_j(\alpha)$, where $j$ belongs to the state space, are the distribution functions associated with the chain states, respectively. Note that each state determines the distribution of $X$ for the current time transition. This means that given the chain process $\{J_n\}$, the random variables of the $\{C_n\}$ process are conditionally independent. If the initial distribution of the chain is the stationary distribution, i.e $P(J_0=i)=\pi_i$ for all $i$ in the finite state space, where $\pi$ is the stationary distribution, then the distribution of $C_n$ is $H(x)=\Sigma \pi_i H_i(x)$. In \cite{denzel1975limit}, it was stated that every stationary chain-dependent process is strongly mixing. We will show that the $J-C$ process, as defined in this paper, is indeed strongly mixing by definition \ref{def-StrongMixing}. 
    \begin{lemma}\label{lem-C_n is strongly mixing}
        $\{C_n\}$ is strongly mixing with $g(k)=\sum_{j} \pi_i \mid P_{ij}^k -\pi_j \mid$, where $\pi$ is the stationary distribution of the chain and $P_{ij}^k=(P^k)_{ij}$.
    \end{lemma}
    \begin{proof}
        Let $A$ and $B$ be as definition \ref{def-StrongMixing}. Then
        
        \ifdouble
\else
        \begin{equation*}
            \begin{aligned}
                 & \mid P(A \cap B)-P(A)P(B) \mid  \\
                 & \leq \sum_{i,j\in\{0,1\}} \mid P(A \cap B,J_m=i,J_{m+k}=j)-P(A,J_m=i)P(B,J_{m+k}=j) \mid \\
                 &=\sum_{i,j\in\{0,1\}} P(A|J_m=i)P(B|J_{m+k}=j)P(J_m=i)\mid P(J_{m+k}=j|J_m=i)-P(J_{m+k}=j) \mid \\
                 &\leq  \sum_{i,j\in\{0,1\}} \pi_i \mid P_{ij}^k -\pi_j \mid = g(k).\\
            \end{aligned}
        \end{equation*}
        \fi
        For each $i \in {0,1}$, $\sum_{j} \pi_i \mid P_{ij}^k -\pi_j \mid \rightarrow 0$ as $k \rightarrow \infty$. This was also used in \cite{o1974limit} in order to show the strongly mixing property, nevertheless, it is easy to notice that ,in fact, as $k \rightarrow \infty$, $(P^k)_{ij} \rightarrow \pi_j$ regardless of $i$.
    \end{proof}
    
    Since strong mixing holds for our sequence of channel capacities, due to the fact that it is a weaken case, so does the condition $D(u_n)$ hold. We would like to show that condition $D'(u_n)$ also holds, so we are able to characterize the rate of exceedance over the threshold $u_n$. Note that we are only interested in a sequence of reals $\{u_n\}$ which satisfies $1-H(u_n)=\uptau/n+o(1/n)$ when considering the $D'(u_n)$ condition. We thus have the following Lemma.
    \begin{lemma}\label{lem-condition D'(u_n) holds on C_n}
           The local dependence condition $D'(u_n)$ holds for the sequence $\{C_n\}$ as defined in \eqref{equ-UserCapacityProcess}.
    \end{lemma}
    \begin{proof}

        \begin{equation*}
            \begin{aligned}
                &\limsup_{n \rightarrow \infty} n\sum_{r=2}^{ \lfloor n/k\rfloor } P(C_1>u_n,C_r>u_n) \\
                &\overset{(a)}{\leq} \limsup_{n \rightarrow \infty} n\sum_{r=2}^{ \lfloor n/k\rfloor } \sum_{i,j\in\{0,1\}} P(C_1>u_n | J_1=i)\\
                &\quad \quad \quad \cdot P(C_r>u_n | J_r=j)  P(J_1=i)P(J_r=j | J_1=i) \\
                &= \limsup_{n \rightarrow \infty} n\sum_{r=2}^{ \lfloor n/k\rfloor } \sum_{i,j\in\{0,1\}} (1-H_i(u_n))(1-H_j(u_n))\pi_i P_{ij}^r\\
                &\overset{(b)}{\leq} \limsup_{n \rightarrow \infty} n\sum_{r=2}^{ \lfloor n/k\rfloor } \sum_{i,j\in\{0,1\}} \left(\frac{\uptau}{n\pi_i}+o\left(\frac{1}{n}\right) \right)
                \left(\frac{\uptau}{n\pi_j}+o\left(\frac{1}{n}\right) \right) \pi_i P_{ij}^r\\
                &\leq \limsup_{n \rightarrow \infty} n\sum_{r=2}^{ \lfloor n/k\rfloor } \sum_{i,j\in\{0,1\}}\frac{1}{\pi_j\pi_i} \left(\frac{\uptau}{n}+o\left(\frac{1}{n}\right) \right)^2  \pi_i P_{ij}^r\\
                &\leq \limsup_{n \rightarrow \infty} n \left\lfloor \frac{n}{k}\right\rfloor \left(\frac{\uptau}{n}+o\left(\frac{1}{n}\right) \right)^2 \sum_{i,j\in\{0,1\}}\frac{1}{\pi_j\pi_i} \pi_i \max_r\{P_{ij}^r\}\\
                &=(\uptau^2+o(1))^2 \frac{1}{k} \sum_{i,j\in\{0,1\}}\frac{1}{\pi_j\pi_i} \pi_i \max_r\{P_{ij}^r\}  \rightarrow 0 \text{  as  } k \rightarrow \infty.
            \end{aligned}
        \end{equation*}
        In the above chain, (a) is since once $J_1$ is known, $C_1$ and $C_r,J_r$ are independent, then we conditioned on $J_r$. (b) is true since $u_n$ maintains $n(1-F_1(u_n))=\uptau +o(1)$, so we have

        \begin{equation*}
         \begin{aligned}
              &1-F_1(u_n)=1-H(u_n)=1-\sum_{i} \pi_i H_i(u_n)\\
              &=\sum_{i} \pi_i(1- H_i(u_n))= \frac{\uptau}{n}+o\left(\frac{1}{n}\right)
         \end{aligned}
        \end{equation*}
        and therefore

        \begin{equation*}
              \pi_l(1-H_l(u_n))=\frac{\uptau}{n}-\sum_{i,i\neq l} \pi_i(1- H_i(u_n))\leq \frac{\uptau}{n} +o\left(\frac{1}{n}\right).
        \end{equation*}
        Note that $\pi_l$ and $\uptau$ are constants.
     \end{proof}
     Thus, in our paradigm, a single users' channel capacity process has the same laws of convergence as if the sequence $\{C_n\}$ was \textit{i.i.d.} with a marginal distribution $H(x)$. Therefore, since the users are independent and each user sees the same marginal distribution $H(x)$, we can analyze the point process of the sequence of all users' capacities at a specific time (e.g, time slot), resulting in the basic case of \textit{i.i.d.} random variables, as given in Theorem  \ref{thm-Point process convergance of iid}. Considering the above, we can now turn to evaluate the expected channel capacity.

     \subsubsection{Distributed Algorithm}
  	Given the number of users, we set a capacity threshold $u$ such that only a small fraction of the users will exceed it. At the beginning of each slot, each user estimates its capacity for that slot. If the capacity anticipated by the user is greater than the capacity threshold, it transmits in that slot. Otherwise, it keeps silent.
    The threshold value is set such that one user exceeds the threshold on average, and as a result, its transmission is successful. We treat this slot as a utilized slot. Hence the expected channel capacity has the form:
     \begin{equation*}
        C_{av}(u)=P_r(\text{utilized slot})E[C|C>u].
     \end{equation*}

     For the calculation of $E[C|C>u]$, the expected capacity experienced by a user who passed $u$, one needs to evaluate the distance of the exceeding points from the threshold.
     \cite{kampeas2014capacity} gives the analytical tools to compute the tail distribution of the exceeding points. These points follow the generalized Pareto distribution, hence by using the PPA exceedance rate results, and since we already showed that we have the same exceedance as an \textit{i.i.d.} case with $H(x)$ as the marginal distribution, the result is the same and we have
     \begin{equation*}
        E[C|C>u]=u+\frac{1}{a_K}+o\left(\frac{1}{a_K}\right),
     \end{equation*}
     where $a_K$ is the normalizing constant as in Theorem \ref{thm-Capacity distribution in a time dependent channel convergence to a Gumbel}.

     We say that a slot is utilized if only one point out of all $K$ points exceeds the threshold. Hence, as $K \to \infty$
     \begin{equation*}
        K\left(\frac{1}{K}\right)\left(1-\frac{1}{K}\right)^{K-1}\to e^{-1},
     \end{equation*}
     where the threshold $u$ was chosen such that $1-H(u)=1/K$. We will elaborate on this value in the next subsection. The expected channel capacity is thus
     \begin{equation}\label{equ-Expected channel capacity (PPA)}
        C_{av}(u)=e^{-1}\left(u+\frac{1}{a_K}+o\left(\frac{1}{a_K}\right)\right).
     \end{equation}
     Clearly, in order to assess the expression above and understand how the capacity scales in a distributed algorithm, one has to compute the value of the threshold.

     \subsubsection{Threshold Estimation}

     The threshold is set such that only one user on average exceeds it in each time slot. This selection maximizes the probability of a successful transmission in a slot, and is asymptotically optimal as $n \rightarrow \infty$ \cite{qin2003exploiting}. By using this rule we can estimate the optimal threshold value for the time dependent channel as well:
    \begin{equation*}
        1-H(u)=\frac{1}{K} \ \  \text{Hence,} \ \ \  1-pF_g(u)-qF_b(u)=\frac{1}{K}.\\
     \end{equation*}

     where the above represents the probability that a user capacity will exceed the threshold $u$. With similar derivation as in Theorem \ref{thm-Capacity distribution in a time dependent channel convergence to a Gumbel} we get that
     
     \ifdouble
\else
     \begin{equation}\label{equ-Estimated threshold}
            u=\sigma_g\sqrt{2\log{K}}\left( 1-\frac{\frac{1}{2}\log{\frac{4\pi}{p^2}}+\frac{1}{2}\log{\log{K}}}{2\log{K}}+o\left( \frac{1}{\log{K}} \right) \right)+\mu_g=b_K.
     \end{equation}
     \fi
     
     Putting \eqref{equ-Estimated threshold} in \eqref{equ-Expected channel capacity (PPA)} we get the same expression for the expected channel capacity under distributed scheduling which show that both approaches have the same scaling laws, the distributed approach being smaller only by a factor of $e^{-1}$, so there was no loss of optimality due to the distributed algorithm.

     \subsubsection{Threshold exceedance process} \label{subsec-Threshold_exceedance_process}

    We saw that the point process described earlier result in convergence of the points exceeding $u_n$ to a non-homogeneous Poisson process $N$. This justify our initial assumption, in the performance part of this work, that the time between threshold exceedances is exponentially distributed with parameter $\uptau$.  It is important to note that convergence in distribution happens if we change the "time scale" by a factor of $n$ as defined in the definition of $N_n$ in \eqref{equ-sequence of point processes N_n}. Meaning, we have a point process in the interval $(0,1]$ and the exceedance of such points has a limiting Poisson distribution.  Of course when considering random arrivals  it would interest us to assess the convergence of the exceeding points for the entire positive line and not just the unit interval.

    Here, we shall use the notation given in \cite{EVT:Springer1983} for the Poisson properties of exceedance.\\
    Let $C_1,C_2,...,C_n$, a sequence of \textit{i.i.d.} r.v.s with distribution $F$, be the channel capacity draws of a specific user. From \cite[Theorem 2.1.1]{EVT:Springer1983}, if $u_n$ satisfies $n(1-F(u_n))\rightarrow \uptau$, then for $k=0,1,2,..., \ \ P(N_n\leq k)\rightarrow e^{-\uptau}\sum_{s=0}^{k}\frac{\uptau^{s}}{s!},$
    where $N_n$ is the number of exceedances of a level $u_n$ by $\{C_n\}$. 
    
     By \cite[Theorem $5.2.1 (\rmnum{2})$]{EVT:Springer1983}, if for each $\uptau >0$, there exists a sequence $\{u_n(\uptau)\}$ satisfying $n(1-F(u_n(\uptau)))=nP(C_1>u_n(\uptau))\rightarrow \uptau$ as $n \rightarrow \infty$, and that $D(u_n(\uptau)), D'(u_n(\uptau))$ hold for all $\uptau > 0$, then for any fixed $\uptau$, $N_n$ converges in distribution to a Poisson process $N$ on $(0,\infty)$ with parameter $\uptau$.

     Clearly, for the approximation model given in Section \ref{Approximate model 2}, the dependence conditions hold, since the sequence $\{C_n\}$ is \textit{i.i.d.}, so one only needs to show that the first condition holds for each $\uptau$. As for the approximation model given in Section \ref{Approximate model 3}, we already showed that the dependence conditions holds and similar derivation will show that  the first condition holds also for each $\uptau$.

    \begin{lemma}\label{lem-condition for poisson convergence on all positive line}
        Assume F is the Gaussian distribution and let $a_n$ and $b_n$ be given according to \cite[Theorem 1.5.3]{EVT:Springer1983}.  Fix any  $\uptau>0$ and set $u_n(\uptau)= \frac{\log{1/\uptau}}{a_n}+b_n$ Then,
        \begin{equation*}
          \lim_{n \to \infty} n(1-F(u_n(\uptau))) = \uptau.
        \end{equation*}

    \end{lemma}
    \begin{proof}
        In a similar way for the derivation of the normalizing constant in \cite[Theorem 1.5.3]{EVT:Springer1983}.  Let us find $u_n(\uptau)$ which satisfies the equivalence condition for the convergence of the expression $n(1-F(u_n(\uptau)))$

        \begin{equation*}
          \begin{aligned}
                    n(1-F(u_n(\uptau))) &\rightarrow \uptau \qquad \text{as  } n\rightarrow \infty \\
                    \frac{nf(u_n(\uptau))}{u_n(\uptau)}&\rightarrow \uptau \qquad  \text{as  } n\rightarrow \infty
          \end{aligned}
        \end{equation*}
        where the second line is true due to the Gaussian relation $1-\Phi(u)\sim \frac{\phi(u)}{u}$ as $u\rightarrow \infty$, which in our case $u_n(\uptau)$ grows with $n$. So

        \begin{equation*}
          \begin{aligned}
                    \frac{1}{\sqrt{2\pi}}e^{-\frac{u_n^2(\uptau)}{2}} &\underset{n\rightarrow \infty}{\rightarrow} \frac{\uptau \ u_n(\uptau)}{n} \\
                    -\log\sqrt{2\pi}-\frac{u_n^2(\uptau)}{2}&\underset{n\rightarrow \infty}{\rightarrow} \log \uptau + \log(u_n(\uptau)) - \log n \qquad
          \end{aligned}
        \end{equation*}
        we know that $\log(u_n(\uptau))=\frac{1}{2}(\log 2 +\log{\log{n}})+o(1) $, hence
        
       \ifdouble
\else
        \begin{equation*}
          \begin{aligned}
                    &\frac{u_n^2(\uptau)}{2}= \log{ \frac{1}{\uptau}}-\frac{1}{2}\log{4\pi}-\frac{1}{2}\log{\log{n}} + \log n +o(1)\\
                    &u_n^2(\uptau)=2\log n\left( 1+ \frac{\log{ \frac{1}{\uptau}}-\frac{1}{2}\log{4\pi}-\frac{1}{2}\log{\log{n}}}{\log n}+o\left(\frac{1}{\log n} \right) \right)\\
                    &u_n(\uptau)=\sqrt{2\log n}\left( 1+ \frac{\log{ \frac{1}{\uptau}}-\frac{1}{2}\log{4\pi}-\frac{1}{2}\log{\log{n}}}{2\log n}+o\left(\frac{1}{\log n} \right) \right)\\
                    &u_n(\uptau)=\frac{\log{1/\uptau}}{a_n}+b_n
          \end{aligned}
        \end{equation*}
        \fi
        where the penultimate line is due to Taylor expansion.
    \end{proof}
    Now, since all conditions for convergence hold, and the exceeding points indeed converge to a Poisson process on the real line of the exceeding points, we can conclude that a user attempts transmission at a rate of $\uptau$, assuming he has packets to send. \ifdouble
\else Figure \ref{fig-PoissonConvergence} show simulation for convergence on the real line. \fi


\ifdouble
\else    
 \begin{figure} \centering
\begin{tikzpicture}[      
        every node/.style={anchor=south west,inner sep=0pt},
        x=1mm, y=1mm,
      ]   
     \node (fig1) at (0,0)
       {\includegraphics[width=0.45\textwidth]{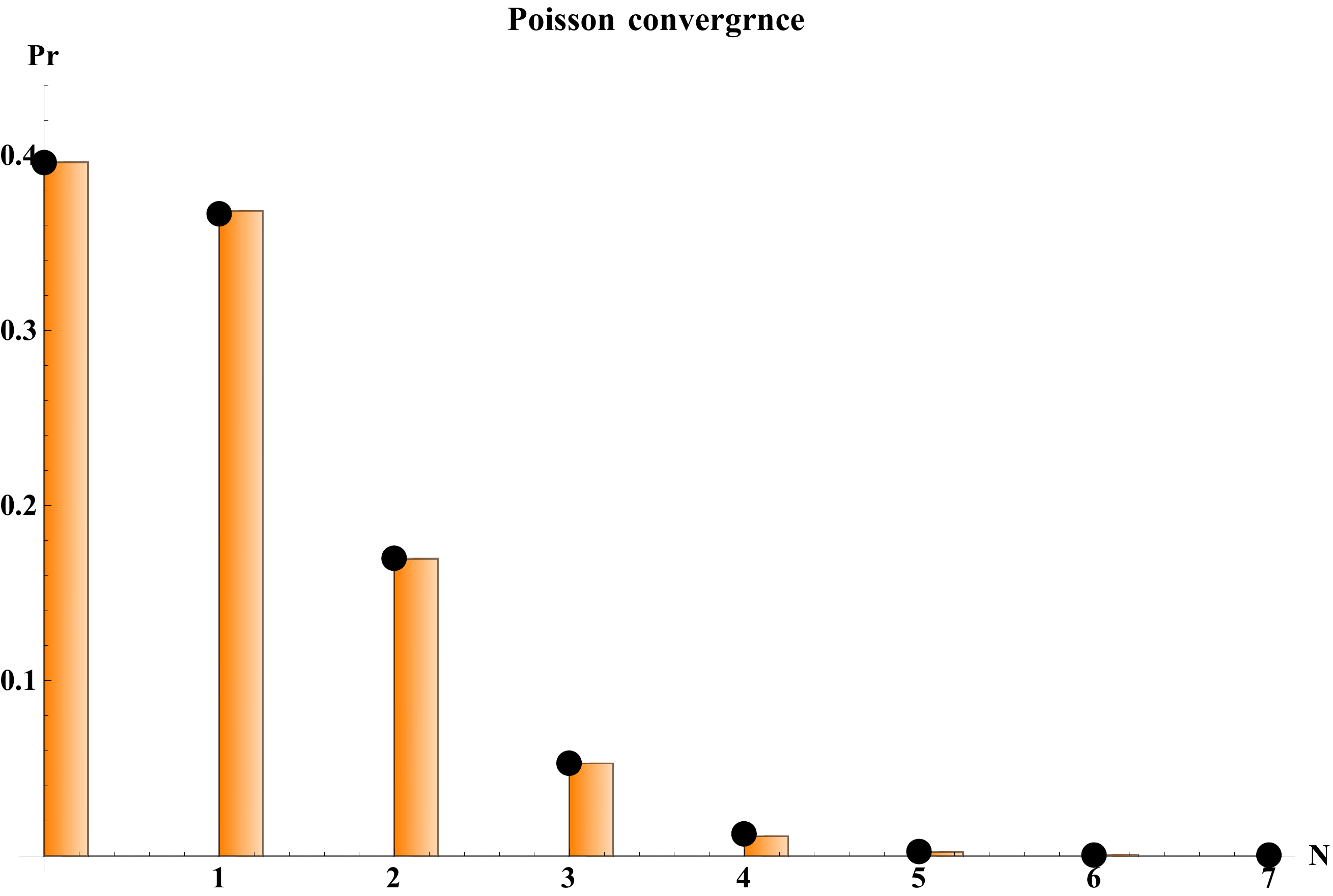}};
     \node (fig2) at (45,20)
       {\tiny
       \begin{tabular}{l|c c}
          \hline
           & Simulation & Poisson \\
           &            & distribution \\
          \hline
          Pr(N=0) & 0.3959 & 0.3961 \\
          Pr(N=1) & 0.3682 & 0.3668 \\
          Pr(N=2) & 0.1696 & 0.1698 \\
          Pr(N=3) & 0.0527 & 0.0524 \\
          Pr(N=4) & 0.0112 & 0.0121 \\
          Pr(N=5) & 0.002  & 0.0022 \\
          Pr(N=6) & 0.0004 & 0.0003 \\
          Pr(N=7) & 0.0    & 0.00004 \\
          \hline
        \end{tabular}};  
\end{tikzpicture}
\caption{Simulation for 10000 user's capacities which follows an \textit{i.i.d.} Gaussian distribution, showing the behaviour of exceedance which converge to a Poisson distribution (the dots) with parameter $\uptau$. On the right up side the values table is given.}
 \label{fig-PoissonConvergence}
\end{figure}
\fi

\section{Conclusion}\label{sec-conclusion}
In this work, we investigated the performance and the channel capacity of a multi-user MAC system in a time dependent and independent environment under distributed scheduling. Specifically, the performance of the system was derived while considering queueing theory aspects. In fact, a precise characterization is a very difficult mission, which up until today was not solved. Therefore, we presented approximation models to describe its behaviour. First, we addressed the \textit{i.i.d.}\ case, where the users do not experience a time varying channel. For that case, we elaborated an existing model and showed results for our paradigm. In addition, we gave another approach, which assumed the queues are independent, derived the probability of collision in the random access mechanism, and enabled us to consider them each as a much more simple queue.  Then, we suggested a queue model, which is time dependent and modelled by our Good-Bad channel model. We showed good agreement between the analytic models and the simulation results. Lastly, the expected channel capacity gain was derived in the case where the dependent capacity sequence was modelled as a stationary process, characterized by the Good-Bad channel Markov process.

\appendices
\section{System-status chain transition probabilities}\label{Appendix A}

    Here we show the calculation of the transition probabilities of the system chain. We use the guide lines in \cite{ephremides1987delay} to simplify things and use also their notions but under our case which state that $r_i=s_i=q_i=p_i$. Due to this fact there are many more possible transition is the state space and therefore the following calculations would be more complicated. In the same manner the calculations are based on the change in the number of active users, thus the transitions can be classified into four types as follows:
    \begin{enumerate}
      \item \underline{The number of active users transits from 0 to 1:}\\
      Only a blocked user may become active, and only one, say $j$. All other blocked blocked must not exceed the threshold. In addition all idle users which may receive a package must not exceed the threshold and therefore, say $w$, becomes blocked. The latter will repeat it self during this calculation. The probability of such a transition is expressed by:
      \begin{equation*}
        P(\Delta A=1, \Delta B=-1+w, \Delta I=-w)=\prod^{K-n-w-1}_{k}\overline{\lambda}_k\prod^{w}_{l}\lambda_l\overline{p}_l
        \prod^{n+1}_{i}\overline{p}_i\frac{p_j}{\overline{p}_j}\left[(1-P_j(0|2))+\lambda_jP_j(0|2)\right]
      \end{equation*}
      where the product is on the sets of users by their transition and $n$ is the number of blocked users after the transition. Note that the user $j$ will be active if his queue is not empty after the successful transmission or a package arrived in the beginning of the slot. Also note that $\overline{p}$ and $\overline{\lambda}$ stand for $1-p$ and $1-\lambda$, respectively.

      \item \underline{The number of active users remains 1:}\\
      This case divides to two subcases, the same active user $j$ remains active or another blocked user $s$ becomes active. Nevertheless $w$ idle users may still become blocked.
      \begin{enumerate}
        \item the probability for the first case is:
        \begin{equation*}
            P(\Delta A=0, \Delta B=+w, \Delta I=-w,j\rightarrow j)= \prod^{K-n-w-1}_{k}\overline{\lambda}_k\prod^{w}_{l}\lambda_l\overline{p}_l
            \prod^{n+1}_{i}\overline{p}_i\frac{p_j}{\overline{p}_j}\left[P_j(1|1)+\lambda_j(1-P_j(1|1))\right]
      \end{equation*}
        \item the probability for the second case is:
        \begin{equation*}
            P(\Delta A=0, \Delta B=+w, \Delta I=-w,j\rightarrow s)= \\ \prod^{K-n-w-1}_{k}\overline{\lambda}_k\prod^{w}_{l}\lambda_l\overline{p}_l
            \prod^{n+1}_{i}\overline{p}_i\frac{p_s}{\overline{p}_s}\left[P_s(1|1)+\lambda_s(1-P_s(1|1))\right]
      \end{equation*}
      \end{enumerate}

      \item \underline{The number of active users transits from 1 to 0:}\\
      This case is subdivided to three subcases,
      \begin{enumerate}
        \item The first subcase is that the active user $j$ becomes idle:
        \begin{equation*}
            P(\Delta A=-1, \Delta B=+w, \Delta I=-w)= \prod^{K-n-w-1}_{k}\overline{\lambda}_k\prod^{w}_{l}\lambda_l\overline{p}_l
            \prod^{n+1}_{i}\overline{p}_i\frac{p_j}{\overline{p}_j}\left[\overline{\lambda}_j(1-P_j(1|1))\right]
        \end{equation*}
        \item The second subcase is that the active user $j$ becomes blocked and blocked user $s$ becomes idle:
        \begin{equation*}
            P(\Delta A=-1, \Delta B=+1-1+w, \Delta I=+1-w)= \prod^{K-n-w-1}_{k}\overline{\lambda}_k\prod^{w}_{l}\lambda_l\overline{p}_l
            \prod^{n+1}_{i}\overline{p}_i\frac{p_s}{\overline{p}_s}\left[\overline{\lambda}_s(1-P_s(1|1))\right]
        \end{equation*}
        \item The third subcase is that the active user $j$ becomes blocked:
        \begin{equation*}
        \begin{aligned}
            &P(\Delta A=-1, \Delta B=+1+w, \Delta I=-w)= \\
            &\prod^{K-n-w-1}_{k}\overline{\lambda}_k\prod^{w}_{l}\lambda_l\overline{p}_l\prod^{n+1}_{i}\overline{p}_i\\
            +&\prod^{K-n-w-1}_{k}\overline{\lambda}_k\prod^{w}_{l}\lambda_l\overline{p}_l \sum^{K-n-w-1}_{s}\frac{\lambda_sp_s}{\overline{\lambda}_s}  \prod^{n+1}_{i}\overline{p}_i\\
            +&\prod^{K-n-w-1}_{k}\overline{\lambda}_k\prod^{w}_{l}\lambda_l\left( 1-\prod^{w}_{s}\overline{p}_s \right)\left(1-\prod^{n}_{i}\overline{p}_i\right)\\
            +&\prod^{K-n-w-1}_{k}\overline{\lambda}_k\prod^{w}_{l}\lambda_l\overline{p}_l\left[ 1-\prod^{n}_{s}\overline{p}_s - \overline{p}_j\prod^{n}_{i}\overline{p}_i\sum^{n}_{q}\frac{p_q}{\overline{p}_q} \right]\\
            +&\prod^{K-n-w-1}_{k}\overline{\lambda}_k\prod^{n}_{i}\overline{p}_i\prod^{w}_{l}\lambda_l\left[ 1-\prod^{w}_{s}\overline{p}_s - \overline{p}_j\prod^{w}_{b}\overline{p}_b\sum^{w}_{q}\frac{p_q}{\overline{p}_q} \right]
        \end{aligned}
        \end{equation*}
        The first expression is for the case which non of the users exceeds the threshold. The second is for the case which one idle user managed to successfully transmits. The third case describes the case which one or more from the blocked and the idle groups tries to transmit, therefore what happens with user $j$ is meaningless. The forth expression describes the situation which user $j$ collided with one or more users from the blocked group, or collision happened between the blocked users and $j$ didn't exceed the threshold. The fifth expression is the same as the fourth concerning the group of idle users which receives a package.
      \end{enumerate}

      \item \underline{The number of active users remains 0 :}\\
      This case is subdivided to three subcases,
      \begin{enumerate}
        \item The first subcase is that all users maintain their status without change:
        \begin{equation*}
            P(\Delta B=0, \Delta I=0)= \left[ 1-\prod^{n}_{i}\overline{p}_i\sum^{n}_{q}\frac{p_q}{\overline{p}_q} \right]\prod^{K-n}_{k}\overline{\lambda}_k + \prod^{n}_{i}\overline{p}_i\prod^{K-n}_{k}\overline{\lambda}_k \sum^{K-n}_{q}\frac{\lambda_qp_q}{\overline{\lambda}_q}
        \end{equation*}
        where the first term is the probability that no packet arrives at the idle users while no blocked user, or at least two blocked users, transmit, and the second term is the probability that no blocked user exceeds the threshold while only one of the idle users successfully transmits.
        \item The second subcase is that one blocked user $j$ becomes idle:
        \begin{equation*}
            P(\Delta B=-1+w, \Delta I=+1-w)= \prod^{K-n-w-1}_{k}\overline{\lambda}_k\prod^{w}_{l}\lambda_l\overline{p}_l
            \prod^{n+1}_{i}\overline{p}_i\left[\overline{\lambda}_j\frac{p_j}{\overline{p}_j}P_j(0|2)\right]
        \end{equation*}
        \item The third subcase describes some situation for $w$ idle users becomes blocked:
        \begin{equation*}
        \begin{aligned}
            &P(\Delta B=+w, \Delta I=-w)= \\
            &\prod^{K-n-w}_{k}\overline{\lambda}_k\prod^{w}_{l}\lambda_l\overline{p}_l\left[ 1-\prod^{n}_{i}\overline{p}_i\sum^{n}_{q}\frac{p_q}{\overline{p}_q} \right]\\
            +&\prod^{K-n-w}_{k}\overline{\lambda}_k\prod^{w}_{l}\lambda_l\left[ 1-\prod^{w}_{s}\overline{p}_s -\prod^{w}_{b}\overline{p}_b\sum^{w}_{q}\frac{p_q}{\overline{p}_q} \right]\\
            +&\prod^{K-n-w}_{k}\overline{\lambda}_k\prod^{w}_{l}\lambda_l\overline{p}_l\sum^{w}_{q}\frac{p_q}{\overline{p}_q} \left[ 1-\prod^{n}_{i}\overline{p}_i\right]\\
            +&\prod^{K-n-w}_{k}\overline{\lambda}_k\sum^{K-n-w}_{s}\frac{\lambda_sp_s}{\overline{\lambda}_s} \prod^{w}_{l}\lambda_l\overline{p}_l\prod^{n}_{i}\overline{p}_i
        \end{aligned}
        \end{equation*}
        The first expression is for the case which non of the $w$ users exceeds the threshold and no blocked user, or at least two blocked users, transmit, the second is for the case which at least two users from $w$ exceeds the threshold, the third case describes collision between one of the users from $w$ with at least one from the blocked users and the forth expression describes the situation which one idle user succussed to transmit while all the other does't exceed the threshold.
      \end{enumerate}
    \end{enumerate}
    According to these transition probabilities we can calculate the steady state of the chain given the auxiliary quantities $p_i(1|1)$ and $p_i(0|2)$, the probability of exceedance and the users arrival rates.


\section{Average success probabilities}\label{Appendix B}

    Once the steady state of the system chain is known, the average success probabilities can be calculated. We do so by calculate the values of $P_B(i),P_A(i)$ and $P_I(i)$ in the same manner as calculated in \cite{ephremides1987delay} while taking in consideration that user must exceed the threshold in order to transmit.
    \begin{enumerate}
      \item For Blocked user $i$:
      \begin{equation*}
        \begin{aligned}
        P_B(i)&=Pr(\text{user $i$ success $\mid$ user $i$ is blocked})
              &=\frac{Pr(\text{user $i$ success and user $i$ is blocked})}{Pr(\text{user $i$ is blocked})}
        \end{aligned}
      \end{equation*}
      where
      \begin{equation*}
        Pr(\text{user $i$ success and user $i$ is blocked})=p_i\sum_{\substack{S_i=2 \\ j\neq i \\ S_j=0,1,2}}P(S_1,...,S_K)\prod_{j\neq i} (\lambda_j\overline{p}_j+\overline{\lambda}_j)^{\delta_{S_j=0}}(\overline{p}_j)^{\delta_{S_j=1}} (\overline{p}_j)^{\delta_{S_j=2}}
      \end{equation*}
      The exponent $\delta$ equals $1$ when it's condition holds.
      \begin{equation*}
        Pr(\text{user $i$ is blocked})=\sum_{S_i=2}P(S_1,..,S_i,..,S_K)
      \end{equation*}

      \item For Active user $i$:
      \begin{equation*}
        \begin{aligned}
        P_A(i)&=Pr(\text{user $i$ success $\mid$ user $i$ is active})
              &=\frac{Pr(\text{user $i$ success and user $i$ is active})}{Pr(\text{user $i$ is active})}
        \end{aligned}
      \end{equation*}
      where
      \begin{equation*}
        Pr(\text{user $i$ success and user $i$ is active})=p_i\sum_{\substack{S_i=1 \\ j\neq i \\ S_j=0,2}}P(S_1,...,S_K)\prod_{j\neq i} (\lambda_j\overline{p}_j+\overline{\lambda}_j)^{\delta_{S_j=0}}(\overline{p}_j)^{\delta_{S_j=2}}
      \end{equation*}
      the difference here is due to the fact that no more than one user may be active.
      \begin{equation*}
        Pr(\text{user $i$ is active})=\sum_{S_i=1}P(S_1,..,S_i,..,S_K)
      \end{equation*}

      \item For Idle user $i$:
      \begin{equation*}
        \begin{aligned}
        P_I(i)&=Pr(\text{user $i$ success $\mid$ user $i$ is idle})
              &=\frac{Pr(\text{user $i$ success and user $i$ is idle})}{Pr(\text{user $i$ is idle})}
        \end{aligned}
      \end{equation*}
      where
      \begin{equation*}
        Pr(\text{user i success and user $i$ is idle})=\lambda_ip_i\sum_{\substack{S_i=0 \\ j\neq i \\ S_j=0,1,2}}P(S_1,...,S_K)\prod_{j\neq i} (\lambda_j\overline{p}_j+\overline{\lambda}_j)^{\delta_{S_j=0}}(\overline{p}_j)^{\delta_{S_j=1}} (\overline{p}_j)^{\delta_{S_j=2}}
      \end{equation*}
      and
      \begin{equation*}
        Pr(\text{user $i$ is idle})=\sum_{S_i=0}P(S_1,..,S_i,..,S_K)
      \end{equation*}
    \end{enumerate}
    After knowing the these values the boundary conditions can be calculated:
    \begin{equation*}
      \begin{aligned}
        P_i(1\mid 1)=&1-\frac{\pi(1,1)}{G^i_1(1)-\pi(1,0)} \\
        P_i(0\mid 2)=&\frac{\pi(0,0)}{G^i_0(1)},
      \end{aligned}
    \end{equation*}
    
    where,
    
     \begin{equation}\label{equ-probability for blocked and empty}
      \pi(0,0)=\frac{\lambda_i\overline{P}_I(i)}{\lambda_iP_A(i)+\overline{\lambda}_iP_B(i)}\pi(1,0),
    \end{equation}
    \begin{equation}\label{equ-probability for idle and empty}
      \pi(1,0)=\frac{\overline{\lambda}_iP_B(i)-\lambda_i\overline{P}_A(i)}{\overline{\lambda}_iP_B(i)-\lambda_i(P_I(i)-P_A(i))},
    \end{equation}
    \begin{equation}\label{equ-probability for active and not empty}
      \pi(1,1)=\frac{\lambda_i}{\overline{\lambda}_i}\pi(0,0),
    \end{equation}
    \begin{equation}\label{equ-probability to be blocked}
      G_{0}^i(1)=\frac{\lambda_i\overline{\lambda}_i\overline{P}_I(i)}{\overline{\lambda}_iP_B(i)-\lambda_i(P_I(i)-P_A(i))},
    \end{equation}
    \begin{equation}\label{equ-probability to be unblocked}
      G_{1}^i(1)=\lambda_i+\overline{\lambda}_i\frac{\overline{\lambda}_iP_B(i)-\lambda_i\overline{P}_A(i)}{\overline{\lambda}_iP_B(i)-\lambda_i(P_I(i)-P_A(i))}.
    \end{equation}

\section{Convergence to Extreme Value Distribution}\label{Appendix C}
        \begin{proof}
        Using the stationary distribution as the marginal distribution does not impairs the convergence when $K \rightarrow \infty$, as shown in  \cite{denzel1975limit}, for such type of dependent sequences. Thus, we can analyze the expected channel capacity using EVT with the distribution above as the marginal distribution function for each user. In order to do so, one first need to prove that convergence to one of the extreme distributions types exists, and derive normilizing constants $a_K$ and $b_K$. The result in \cite[1.6.2]{EVT:Springer1983} gives a necessary and sufficient condition on the marginal distribution $F$ to belong to each of the three possible domains of attraction of the extreme value distributions. 
    
    The first sufficient condition states that if $f$ has a negative derivative $f'$ for all $x$ in some interval $(x_0,x_F),\ (x_F\leq \infty),\ f(x)=0$ for $x\geq x_F$, and
    \begin{equation}\label{equ-Sufficient type 1 condition}
        \lim_{t \uparrow x_F} \frac{f'(t)(1-F(t))}{f^2(t)}=-1
    \end{equation}
    then $F$ is in the domain of attraction of Type \Rmnum{1} extreme value distribution (Gumbel distribution). The second necessary and sufficient condition for $F$ to be in the domain of attraction of Type \Rmnum{1} extreme value distribution, states that there exists some strictly positive function $g(t)$ such that
    \begin{equation}\label{equ-Necessary and sufficient type 1 condition}
     \lim_{t \uparrow x_F} \frac{1-F(t+xg(t))}{1-F(t)}=e^{-x}\\
    \end{equation}
    for all real $x$.\\

           In the following we give compliance for the convergence conditions and the derivation of $a_K$ and $b_K$.
           The stationary distribution  $F(t)=pF_g(t)+qF_b(t)$ as shown earlier has a negative derivative $f'$ from $x_0=\max\{\mu_g,\mu_b\}$ till $\infty$. So we only need to show that \eqref{equ-Sufficient type 1 condition} holds,
           \begin{equation*}
                \begin{aligned}
                    &\lim_{t \rightarrow \infty} \frac{f'(t)\left(1-F(t)\right)}{f^2(t)} =
                    \lim_{t \rightarrow \infty} \frac{\left(pf_g'(t)+qf_b'(t)\right)\left(1-\left(pF_g(t)+qF_b(t)\right)\right)}{\left(pf_g(t)+qf_b(t)\right)^2}\\
                    &=\lim_{t \rightarrow \infty} \frac{\frac{1}{2}\left(pf_g'(t)+qf_b'(t)\right)\left(pErfc\left(\frac{t-\mu_g}{\sqrt{2}\sigma_g}\right) +qErfc\left(\frac{t-\mu_b}{\sqrt{2}\sigma_b}\right)\right)}{\left(pf_g(t)+qf_b(t)\right)^2}\\
                \end{aligned}
               \end{equation*}
               In \cite{abramowitz2012handbook}, 7.1.13, we can find upper and lower bounds for the complementary error function,
               \begin{equation*}
                \frac{2}{\sqrt{\pi}}\frac{e^{-t^2}}{t+\sqrt{t^2+2}}< Erfc(t)\leq  \frac{2}{\sqrt{\pi}}\frac{e^{-t^2}}{t+\sqrt{t^2+\frac{4}{\pi}}}.
           \end{equation*}
           where these inequalities are true for $t>0$ which fits our case for $t \rightarrow \infty$. Using these bounds we will show with the sandwich rule that the limit above converge to $-1$. Let us consider first the lower bound of the complementary error function,
           \begin{equation*}
                 \begin{aligned}
                     &\lim_{t \rightarrow \infty} \frac{\frac{1}{2}\left(pf_g'(t)+qf_b'(t)\right)
                     \left( p\frac{2}{\sqrt{\pi}}\frac{e^{-\frac{(t-\mu_g)^2}{2\sigma_g^2}}}{\frac{t-\mu_g}{\sqrt{2}\sigma_g}+\sqrt{\frac{(t-\mu_g)^2}{2\sigma_g^2}+2}} + q\frac{2}{\sqrt{\pi}}\frac{e^{-\frac{(t-\mu_b)^2}{2\sigma_b^2}}}{\frac{t-\mu_b}{\sqrt{2}\sigma_b}+\sqrt{\frac{(t-\mu_b)^2}{2\sigma_b^2}+2}}\right)}
                     {\left(pf_g(t)+qf_b(t)\right)^2}\\
                     &=  \lim_{t \rightarrow \infty}\frac{1}{\sqrt{\pi}} \frac{\left(pf_g'(t)+qf_b'(t)\right)
                     \left( p\frac{\sqrt{2\pi}\sigma_g f_g(t)}{\frac{t-\mu_g}{\sqrt{2}\sigma_g}+\sqrt{\frac{(t-\mu_g)^2}{2\sigma_g^2}+2}} +
                            q\frac{\sqrt{2\pi}\sigma_b f_b(t)}{\frac{t-\mu_b}{\sqrt{2}\sigma_b}+\sqrt{\frac{(t-\mu_b)^2}{2\sigma_b^2}+2}}   \right)}
                     {\left(pf_g(t)+qf_b(t)\right)^2}\\
                     &=  \lim_{t \rightarrow \infty}-\sqrt{2} \frac{\left(pf_g(t)\frac{t-\mu_g}{\sigma_g^2}+qf_b(t)\frac{t-\mu_b}{\sigma_b^2}\right)
                     \left(p \frac{\sigma_g f_g(t)}{\frac{t-\mu_g}{\sqrt{2}\sigma_g}+\sqrt{\frac{(t-\mu_g)^2}{2\sigma_g^2}+2}} +
                           q \frac{\sigma_b f_b(t)}{\frac{t-\mu_b}{\sqrt{2}\sigma_b}+\sqrt{\frac{(t-\mu_b)^2}{2\sigma_b^2}+2}}   \right)}
                     {\left(pf_g(t)+qf_b(t)\right)^2}\\
                 \end{aligned}
           \end{equation*}
           The Limit above can be break to four different limits,
           \begin{equation*}
                 \begin{aligned}
                     &\lim_{t \rightarrow \infty}-\sqrt{2}p^2 \frac{f_g(t)\frac{t-\mu_g}{\sigma_g^2}
                      \frac{\sigma_g f_g(t)}{\frac{t-\mu_g}{\sqrt{2}\sigma_g}+\sqrt{\frac{(t-\mu_g)^2}{2\sigma_g^2}+2}}}
                     {\left(pf_g(t)+qf_b(t)\right)^2}+
                     \lim_{t \rightarrow \infty}-\sqrt{2}q^2 \frac{f_b(t)\frac{t-\mu_b}{\sigma_b^2}
                      \frac{\sigma_b f_b(t)}{\frac{t-\mu_b}{\sqrt{2}\sigma_b}+\sqrt{\frac{(t-\mu_b)^2}{2\sigma_b^2}+2}}}
                     {\left(pf_g(t)+qf_b(t)\right)^2}+\\
                     &\lim_{t \rightarrow \infty}-\sqrt{2}pq \frac{f_g(t)\frac{t-\mu_g}{\sigma_g^2}
                      \frac{\sigma_b f_b(t)}{\frac{t-\mu_b}{\sqrt{2}\sigma_b}+\sqrt{\frac{(t-\mu_b)^2}{2\sigma_b^2}+2}}}
                     {\left(pf_g(t)+qf_b(t)\right)^2}+
                     \lim_{t \rightarrow \infty}-\sqrt{2}pq \frac{f_b(t)\frac{t-\mu_b}{\sigma_b^2}
                      \frac{\sigma_g f_g(t)}{\frac{t-\mu_g}{\sqrt{2}\sigma_g}+\sqrt{\frac{(t-\mu_g)^2}{2\sigma_g^2}+2}}}
                     {\left(pf_g(t)+qf_b(t)\right)^2}\\
                 \end{aligned}
           \end{equation*}
           please notice that the first and the second limits are similar with the exception of their indexes, and so does the third and the fourth limit. We start with the first limit calculation,
           \begin{equation*}
                 \begin{aligned}
                     &\lim_{t \rightarrow \infty}-\sqrt{2}p^2 \ \frac{f_g(t) \ \frac{t-\mu_g}{\sigma_g^2} \
                      \frac{\sigma_g f_g(t)}{\frac{t-\mu_g}{\sqrt{2}\sigma_g}+\sqrt{\frac{(t-\mu_g)^2}{2\sigma_g^2}+2}}}
                     {\left(pf_g(t)+qf_b(t)\right)^2}
                     =\lim_{t \rightarrow \infty}-\sqrt{2} \ \frac{p^2f_g^2(t) \ (t-\mu_g)\sigma_g}
                     {\left(pf_g(t)+qf_b(t)\right)^2 \ \sigma_g^2 \ \frac{t-\mu_g}{\sqrt{2}\sigma_g}\left(1+\sqrt{1+\frac{4\sigma_g^2}{(t-\mu_g)^2}} \right) }\\
                     &=\lim_{t \rightarrow \infty} -2 \ \frac{p^2f_g^2(t)}
                     {\left(pf_g(t)+qf_b(t)\right)^2 \ \left(1+\sqrt{1+\frac{4\sigma_g^2}{(t-\mu_g)^2}} \right)}
                     =\lim_{t \rightarrow \infty} -2 \ \frac{p^2f_g^2(t)}
                     {\left(pf_g(t)+qf_b(t)\right)^2} \cdot
                     \lim_{t \rightarrow \infty} \frac{1}
                     {\left(1+\sqrt{1+\frac{4\sigma_g^2}{(t-\mu_g)^2}} \right)}\\
                     &=\lim_{t \rightarrow \infty} -2 \ \frac{p^2f_g^2(t)}
                     {\left(pf_g(t)+qf_b(t)\right)^2 } \cdot \frac{1}{2}
                     =\lim_{t \rightarrow \infty} - \ \frac{p^2f_g^2(t)}
                     {\left(pf_g(t)+qf_b(t)\right)^2}
                     =-\left(\lim_{t \rightarrow \infty} \frac{pf_g(t)}
                     {\left(pf_g(t)+qf_b(t)\right)}\right)^2\\
                     &\overset{(a)}{=}-\left(\lim_{t \rightarrow \infty} \frac{1}
                     {1+\frac{qf_b(t)}{pf_g(t)}}\right)^2
                     =-\left( \frac{1}
                     {1+\lim_{t \rightarrow \infty}\frac{qf_b(t)}{pf_g(t)}}\right)^2 \ =
                \end{aligned}
           \end{equation*}
           \begin{equation*}
                \begin{aligned}
                     &\Rightarrow \lim_{t \rightarrow \infty}\frac{qf_b(t)}{pf_g(t)}
                     =\frac{\sigma_g}{\sigma_b} \frac{q}{p} \lim_{t \rightarrow \infty} e^{-\frac{(t-\mu_b)^2}{2\sigma_b^2}+\frac{(t-\mu_g)^2}{2\sigma_g^2}}
                     = \frac{\sigma_g}{\sigma_b}\frac{q}{p} e^{\lim_{t \rightarrow \infty} -\frac{(t-\mu_b)^2}{2\sigma_b^2}+\frac{(t-\mu_g)^2}{2\sigma_g^2}}\\
                     &\Rightarrow \lim_{t \rightarrow \infty}\frac{t^2(\sigma_b^2-\sigma_g^2)+t(2\mu_b\sigma_g^2-2\mu_g\sigma_b^2)+C}{2\sigma_b^2\sigma_g^2}
                     = \left\{
                                \begin{array}{l l}
                                   \infty  & \quad \sigma_g^2 < \sigma_b^2\\
                                   -\infty  & \quad \sigma_g^2 \geq \sigma_b^2 \ \text{assuming} \ \mu_g>\mu_b
                                \end{array} \right.
                \end{aligned}
           \end{equation*}
           So,
           \begin{equation*}
                 \lim_{t \rightarrow \infty}\frac{qf_b(t)}{pf_g(t)} =
                        \left\{
                            \begin{array}{l l}
                               \infty  & \quad \sigma_g^2 < \sigma_b^2\\
                               0  & \quad \sigma_g^2 \geq \sigma_b^2 \ \text{assuming} \ \mu_g>\mu_b
                        \end{array} \right.
           \end{equation*}
           Note that in (a) we assume that $p \neq 0$. This assumption implies that the situation which all the users are in the bad group is not taken in consideration. For that case all the users have the same channel and the analysis is known and not in our interest. In the same manner we assume also that $q \neq 0$ for the opposite situation.
           Hence, the first limit result is
           \begin{equation*}
                 \lim_{t \rightarrow \infty}-\sqrt{2}p^2 \ \frac{f_g(t) \ \frac{t-\mu_g}{\sigma_g^2} \
                  \frac{\sigma_g f_g(t)}{\frac{t-\mu_g}{\sqrt{2}\sigma_g}+\sqrt{\frac{(t-\mu_g)^2}{2\sigma_g^2}+2}}}
                 {\left(pf_g(t)+qf_b(t)\right)^2}=
                        \left\{
                            \begin{array}{l l}
                               0  & \quad \sigma_g^2 < \sigma_b^2\\
                               -1  & \quad \sigma_g^2 \geq \sigma_b^2 \ \ \text{assuming} \ \mu_g>\mu_b
                        \end{array} \right.
           \end{equation*}
           As mentioned earlier the first and the second limits different only in their indexes, therefore the result for the second limit is
           \begin{equation*}
                 \lim_{t \rightarrow \infty}-\sqrt{2}q^2 \frac{f_b(t)\frac{t-\mu_b}{\sigma_b^2}
                  \frac{\sigma_b f_b(t)}{\frac{t-\mu_b}{\sqrt{2}\sigma_b}+\sqrt{\frac{(t-\mu_b)^2}{2\sigma_b^2}+2}}}
                 {\left(pf_g(t)+qf_b(t)\right)^2}=
                        \left\{
                            \begin{array}{l l}
                               -1  & \quad \sigma_g^2 < \sigma_b^2\\
                               0  & \quad \sigma_g^2 \geq \sigma_b^2 \ \ \text{assuming} \ \mu_g>\mu_b
                        \end{array} \right.
           \end{equation*}
           We turn now for the third limit calculation
           \begin{equation*}
                \begin{aligned}
                     &\lim_{t \rightarrow \infty}-\sqrt{2}pq \frac{f_g(t)\frac{t-\mu_b}{\sigma_g^2}
                      \frac{\sigma_b f_b(t)}{\frac{t-\mu_b}{\sqrt{2}\sigma_b}+\sqrt{\frac{(t-\mu_b)^2}{2\sigma_b^2}+2}}}
                     {\left(pf_g(t)+qf_b(t)\right)^2}
                     =\lim_{t \rightarrow \infty}-\sqrt{2} \ \frac{pqf_g(t) \ (t-\mu_g)\sigma_b \ f_b(t) \sqrt{2} \sigma_b}
                     {\left(pf_g(t)+qf_b(t)\right)^2 \ \sigma_g^2 \ (t-\mu_b) \ \left(1+\sqrt{1+\frac{4\sigma_b^2}{(t-\mu_b)^2}} \right) }\\
                     &=-2\frac{\sigma_b^2}{\sigma_g^2}\lim_{t \rightarrow \infty} \ \frac{pf_g(t)}{\left(pf_g(t)+qf_b(t)\right)} \cdot
                     \lim_{t \rightarrow \infty}\ \frac{qf_b(t)}{\left(pf_g(t)+qf_b(t)\right)} \cdot
                     \lim_{t \rightarrow \infty}\ \frac{(t-\mu_g)}{(t-\mu_b)}\cdot
                     \lim_{t \rightarrow \infty}\ \frac{1}
                     {\left(1+\sqrt{1+\frac{4\sigma_g^2}{(t-\mu_g)^2}} \right) }\\
                     &=-\frac{\sigma_b^2}{\sigma_g^2}\lim_{t \rightarrow \infty}\ \frac{pf_g(t)}{\left(pf_g(t)+qf_b(t)\right)} \cdot
                     \lim_{t \rightarrow \infty}\ \frac{qf_b(t)}{\left(pf_g(t)+qf_b(t)\right)}\\
                     &=-\frac{\sigma_b^2}{\sigma_g^2}
                     \left\{
                        \begin{array}{l l}
                             0  & \quad \sigma_g^2 < \sigma_b^2\\
                             1  & \quad \sigma_g^2 \geq \sigma_b^2 \ \ \text{assuming} \ \mu_g>\mu_b
                        \end{array} \right.
                     \cdot\left\{
                        \begin{array}{l l}
                             1  & \quad \sigma_g^2 < \sigma_b^2\\
                             0  & \quad \sigma_g^2 \geq \sigma_b^2 \ \ \text{assuming} \ \mu_g>\mu_b\
                        \end{array} \right.
                     =0
                 \end{aligned}
           \end{equation*}
           The fourth limit also share the same result, and if we add all the parts we can see that the limit converge to $-1$. If we consider the upper bound of the complementary error function we will find also that the limit convergence to $-1$ since the analytical development is the same and the limit
           \begin{equation*}
                 \lim_{t \rightarrow \infty} \frac{1}
                 {\left(1+\sqrt{1+\frac{4\sigma_g^2}{(t-\mu_g)^2}} \right)}
                 =\lim_{t \rightarrow \infty} \frac{1}
                 {\left(1+\sqrt{1+\frac{8\sigma_g^2}{\pi(t-\mu_g)^2}} \right)}
                 =\frac{1}{2}
           \end{equation*}
           Therefore we can conclude that  condition \eqref{equ-Sufficient type 1 condition} holds for our stationary distribution.
            We now show that the second condition \eqref{equ-Necessary and sufficient type 1 condition} also holds for the stationary distribution. Let us examine the expression:
       \begin{equation*}
            1-F(t)=p(1-F_g(t))+q(1-F_b(t))
       \end{equation*}
       Since $F_i(t)$ is a Gaussian distribution we shall use the asymptotic relation:
       \begin{equation}\label{equ-Asymptotic relation of Gaussian distribution}
            1-\Phi(t) \sim \frac{\phi(t)}{t} \quad \text{as } t \rightarrow \infty
       \end{equation}
       \begin{equation*}
            \begin{aligned}
                1-F(t)&=p\left(\frac{\sigma_g}{t-\mu_g}\phi(\frac{t-\mu_g}{\sigma_g}) \right)+q\left(\frac{\sigma_b}{t-\mu_b}\phi(\frac{t-\mu_b}{\sigma_b}) \right)      \\
                      &=\frac{1}{\sqrt{2\pi}}\left(\frac{p}{t-\mu_g}e^{-\frac{(t-\mu_g)^2}{2\sigma_g^2}} +\frac{q}{t-\mu_b}e^{-\frac{(t-\mu_b)^2}{2\sigma_b^2}} \right)  \\
                      &\overset{(a)}{=}\frac{1}{\sqrt{2\pi}}\frac{p}{t-\mu_g}e^{-\frac{(t-\mu_g)^2}{2\sigma_g^2}}(1+o(1)) \quad \text{as } t \rightarrow \infty
            \end{aligned}
       \end{equation*}
       where (a) is true since
       \begin{equation*}
       \lim_{t \rightarrow \infty} \frac{\frac{\frac{q}{t-\mu_b}e^{-\frac{(t-\mu_b)^2}{2\sigma_b^2}}}{\frac{p}{t-\mu_g}e^{-\frac{(t-\mu_g)^2}{2\sigma_g^2}}}}{1}=
       \lim_{t \rightarrow \infty} \frac{e^{-\frac{(t-\mu_b)^2}{2\sigma_b^2}}}{e^{-\frac{(t-\mu_g)^2}{2\sigma_g^2}}} = 0
       \end{equation*}
       assuming $\sigma_g>\sigma_b$.
       So taking in consideration condition \eqref{equ-Necessary and sufficient type 1 condition}:
       \begin{equation*}
            \begin{aligned}
                &\frac{1-F(t+xg(t))}{1-F(t)}=\frac{\frac{1}{\sqrt{2\pi}}\frac{p}{t+xg(t)-\mu_g}e^{-\frac{(t+xg(t)-\mu_g)^2}{2\sigma_g^2}}(1+o(1))}
                                             {\frac{1}{\sqrt{2\pi}}\frac{p}{t-\mu_g}e^{-\frac{(t-\mu_g)^2}{2\sigma_g^2}}(1+o(1))}                     \\
                &=\frac{t-\mu_g}{t+xg(t)-\mu_g} e^{\frac{-(t+xg(t)-\mu_g)^2+(t-\mu_g)^2}{2\sigma_g^2}} (1+o(1))            \\
                &=\frac{1}{1+\frac{xg(t)}{t-\mu_g}}e^{-\frac{g(t)x(t-\mu_g)}{\sigma_g^2}}e^{-\frac{g^2(t)x^2}{2\sigma_g^2}}(1+o(1))=\\
            \end{aligned}
       \end{equation*}
       By choosing $g(t)=\frac{\sigma_g^2}{t-\mu_g}$ as the strictly positive function for $t \rightarrow \infty$ we get
       \begin{equation*}
                =\frac{1}{1+\frac{x\sigma_g^2}{(t-\mu_g)^2}}e^{-x}e^{-\frac{\sigma_g^2 \ x^2}{2(t-\mu_g)^2}}(1+o(1))=e^{-x} \quad \quad \text{as } t \rightarrow \infty
       \end{equation*}
       That conclude that the distribution function $F(x)$ belongs to the domain of attraction of Type \Rmnum{1}. Similar analysis can be found in \cite{mladenovic1999extreme}, where some examples for convergence of sequences of independent random variables with the same mixed distribution is investigated.\\       
       
       We now derive the normalizing constants $a_K$ and $b_K$:\\
       According to EVT results for \emph{i.i.d.} sequences \cite[Theorem 1.5.1]{EVT:Springer1983}, $u_K=u_K(x)=x/a_K+b_K$ is a sequence of real numbers such that $K(1-F(u_K))\rightarrow \uptau$ as $K\rightarrow \infty$, therefore in our case:
       \begin{equation*}
         1-pF_g(u_K)-qF_b(u_K)\rightarrow \frac{1}{K}e^{-x}, \quad K\rightarrow \infty
       \end{equation*}
       where $\uptau=e^{-x}$. The same way as the previous proof using \eqref{equ-Asymptotic relation of Gaussian distribution} we obtain that
       \begin{equation*}
         \left(\frac{p\sigma_g}{u_K-\mu_g}\phi(\frac{u_K-\mu_g}{\sigma_g}) \right)+\left(\frac{q\sigma_b}{u_K-\mu_b}\phi(\frac{u_K-\mu_b}{\sigma_b}) \right)\rightarrow \frac{1}{K}e^{-x}
       \end{equation*}
       \begin{equation*}
         \frac{1}{\sqrt{2\pi}}\frac{p}{u_K-\mu_g}e^{-\frac{(u_K-\mu_g)^2}{2\sigma_g^2}}(1+o(1)) \rightarrow \frac{1}{K}e^{-x}
       \end{equation*}
       the last step is true since $u_K \rightarrow \infty$ as $K \rightarrow \infty$, similar to the pervious proof.
       \begin{equation}\label{equ-Proof for a_n b_n (1)}
                 -\frac{1}{2}\log2\pi+\log p-\log{(u_K-\mu_g)}-\frac{(u_K-\mu_g)^2}{2\sigma_g^2}+\log K+x+o(1) \rightarrow 0
       \end{equation}
       It follows at once that $(t-\mu_g)^2 / 2\log K \rightarrow 1$, and hence
       \begin{equation*}
         \log{(u_K-\mu_g)}=\frac{1}{2}(\log 2 +\log{\log{K}})+o(1)
       \end{equation*}
       Putting this in \eqref{equ-Proof for a_n b_n (1)} ,we obtain
       \begin{equation*}
                 \frac{(u_K-\mu_g)^2}{2\sigma_g^2}=-\frac{1}{2}\log2\pi+\log p-\frac{1}{2}(\log 2 +\log{\log{K}})+\log K +x+o(1)
       \end{equation*}
       or
       \begin{equation*}
                 \frac{(u_K-\mu_g)^2}{\sigma_g^2}=
                  2\log K\left(1+\frac{x-\frac{1}{2}\log{\frac{4\pi}{p^2}}-\frac{1}{2}\log{\log{K}}}{\log K}+o\left(\frac{1}{\log K}\right)\right)
       \end{equation*}
       and hence
       \begin{equation*}
                 \frac{(u_K-\mu_g)}{\sigma_g}=
                  \sqrt{2\log K}\left(1+\frac{x-\frac{1}{2}\log{\frac{4\pi}{p^2}}-\frac{1}{2}\log{\log{K}}}{2\log K}+o\left(\frac{1}{\log K}\right)\right)
       \end{equation*}
       so by using expansion we have,
       \begin{equation*}
                 u_K=\sigma_g\sqrt{2\log K}
                  \left(1+\frac{x-\frac{1}{2}\log{\frac{4\pi}{p^2}}-\frac{1}{2}\log{\log{K}}}{2\log K}+o\left(\frac{1}{\log K}\right)\right)+\mu_g
       \end{equation*}
       since we know that $u_K=x/a_K+b_K$ we conclude that
       \begin{equation*}
           \begin{aligned}
                 &a_K=\frac{\sqrt{2\log{K}}}{\sigma_g}\\
                 &b_K=\sigma_g\left((2\log K)^{1/2}-\frac{\log{\log K}+\log{\frac{4\pi}{p^2}}}{2(2\log K)^{1/2}}\right)+\mu_g
           \end{aligned}
       \end{equation*}
\end{proof}


\bibliography{My_bib}

\end{document}